\documentclass{scrartcl}
\KOMAoptions{paper=letter}

\usepackage[]{group} %
\usepackage{algorithm}
\usepackage{algpseudocode}
\usepackage{thm-restate}
\algblock{Input}{EndInput}
\algnotext{EndInput}

\newcommand{\Desc}[2]{\State \makebox[3em][l]{#1}#2}
\newcommand{\isrsd}{\ensuremath{I_\rsd}\xspace}
\newcommand{\domain}{\ensuremath{\mathcal{R}}\xspace}
\newcommand{\canonical}{\ensuremath{\mathcal{R}^*}\xspace}
\newcommand{\chch}{\ensuremath{\mathcal{R^>}}\xspace}
\newcommand{\Tau}{\ensuremath{\mathcal{T}}\xspace}

\newcommand{\serd}{\ensuremath{\mathit{SD}}\xspace}

\newcommand{\vect}[1]{\ensuremath{\mathbf{#1}}\xspace}

\setcitestyle{authoryear}

\definecolor{grey}{gray}{0.9}

 \title{Towards a Characterization of\\ Random Serial Dictatorship\footnote{In previous versions of this paper, we reported that a computer program based on our techniques has proved a characterization of random serial dictatorship for $n\le 5$. We now have doubts concerning the correcntess of the implementation and have removed any mention of it.}}

 \author{Felix Brandt \quad Matthias Greger \quad Ren\'e Romen\\Technical University of Munich, Germany}

\begin{document}

	\maketitle

	Random serial dictatorship (\rsd) is a randomized assignment rule that---given a set of $n$ agents with strict preferences over $n$ houses---satisfies equal treatment of equals, \emph{ex post} efficiency, and strategyproofness. For $n \le 3$, \citet{BoMo01a} have shown that \rsd is characterized by these axioms. Extending this characterization to arbitrary $n$ is a long-standing open problem. By weakening \emph{ex post} efficiency and strategyproofness, we reduce the question of whether \rsd is characterized by these axioms for fixed $n$ to determining whether a matrix has rank $n^2 n!^n$. 
We 
provide computer-generated counterexamples to show that two other approaches for proving the characterization (using deterministic extreme points or restricted domains of preferences) are inadequate.

	\section{Introduction}

Assigning objects to individual agents is a fundamental problem that has received considerable attention by computer scientists as well as economists \citep[e.g.,][]{CDE+06a,SoUn10a,Manl13a,BCM15a}. The problem is known as the \emph{assignment problem}, the \emph{house allocation problem}, or \emph{two-sided matching with one-sided preferences}. In its simplest form, there are $n$ agents, $n$ houses, and each house needs to be allocated to exactly one agent based on the strict preferences of each agent over the houses.
Applications are diverse and include assigning dormitories to students, jobs to applicants, processor time slots to jobs, parking spaces to employees, offices to workers, etc.

A class of simple, well understood, and often applied deterministic assignment rules are \emph{serial dictatorships}, which are based on a fixed priority order over the agents that is independent of the reported preferences. The agent with the highest priority gets to pick her most preferred house, then the second agent chooses her most preferred among the remaining houses, and so on. Serial dictatorships are guaranteed to return a Pareto efficient allocation. On top of that, they are neutral (when houses are permuted, the assignment is permuted accordingly), nonbossy (an agent cannot affect the assignment to other agents without changing the house allocated to herself), and strategyproof (no agent can misreport her preferences in order to obtain a more preferred house). Unsurprisingly, like any deterministic rule, serial dictatorships are highly unfair. For example, consider two agents who both prefer house $h_1$ to $h_2$. Any deterministic rule strongly discriminates the agent who receives $h_2$.

Fairness is typically established by allowing for \emph{probabilistic} assignment rules where each agent receives each house with some probability and the probabilities sum up to $1$ for each agent and each house. The resulting probability matrix is called a bistochastic matrix.
The Birkhoff-von Neumann theorem shows that every bistochastic matrix can be decomposed into a convex combination of permutations matrices. As a consequence, every probabilistic assignment rule can be implemented in practice by picking a deterministic assignment rule at random. 
The two most prominent probabilistic assignment rules are \emph{random serial dictatorship} (\rsd)---also known as \emph{random priority}---and the \emph{probabilistic serial} rule \citep{BoMo01a}. 	

A natural way to obtain a randomized assignment rule is to apply a deterministic rule to every permutation of the agents' roles and then uniformly randomize over all of these $n!$ deterministic assignments. Such a \emph{symmetrization} ensures that ``equals are treated equally''. In fact, \rsd is defined as the symmetrization of all serial dictatorships and has been shown to be equivalent to the symmetrization of Gale's top trading cycles mechanism \citep{AbSo98a, Knut96a}. \citet{Sven99a} showed that any deterministic, strategyproof, nonbossy, and neutral assignment rule is serially dictatorial, implying that the symmetrization of any such rule has to coincide with \rsd. \citet{Papa00a} and \citet{PyUn17a} have characterized broader classes of deterministic assignment rules by replacing neutrality with efficiency.

An approach that works well in conjunction with symmetrization is to prove that the set of rules satisfying a given set of axioms forms a polyhedron with deterministic extreme points (DEP) \citep[see, e.g.,][]{PyUn15a,GPS17a,RoSa20a}.
As a consequence, all rules that satisfy these axioms can be represented as convex combinations of deterministic rules.
Then, any symmetrization results on those deterministic rules naturally extend to all probabilistic rules.
\citet{PyUn15a} have shown that the set of strategyproof rules does not satisfy the DEP property.

The main axiomatic advantage of \rsd is that it satisfies strategyproofness while also guaranteeing efficiency and fairness to some extent. While \rsd does satisfy \emph{ex post} efficiency, it violates a stronger efficiency notion called ordinal efficiency or \sd-efficiency \citep{BoMo01a}. In fact, \citeauthor{BoMo01a} showed that strategyproofness and equal treatment of equals are incompatible with ordinal efficiency. Furthermore, they observed that \rsd only satisfies a weak notion of envy-freeness. The probabilistic serial rule, on the other hand, satisfies ordinal efficiency and envy-freeness but violates strategyproofness. 

A characterization of \rsd via equal treatment of equals, \emph{ex post} efficiency, and strategyproofness is a long-standing open problem \citep[see, e.g.,][]{PaSe13a,PyTr23a} and would cement its pivotal role in settings where strategyproofness is indispensable.  

Unfortunately, to the best of our knowledge, there does not even exist a characterization of all \emph{deterministic}, strategyproof, and efficient assignment rules
\citep[cf.][]{Sven99a}. %
Furthermore, \citet{ABB13b} and \citet{SaSe15a} showed that it is NP-complete to decide whether an agent receives a given house with positive probability under \rsd, stressing its combinatorial intricacy. 

\citet{PyTr23b} recently showed that \rsd is characterized by symmetry, efficiency, and obvious strategyproofness among all assignment rules that, roughly speaking, can be represented as a symmetrization of an extensive-form game where in each stage, one agent is allowed to pick one house from a subset of the remaining houses or ``pass'' on this opportunity. %
Furthermore, \citet{PyTr23a} point out that equal treatment of equals, \emph{ex post} efficiency, and strategyproofness do not suffice to characterize \rsd when using a stronger equivalence notion that interprets two rules as different if they produce different distributions over deterministic assignments, even when the probabilistic assignment is still the same. By contrast, we consider two rules as equivalent if, for each profile, they return the same probabilistic assignment. 

In this paper, we use a linear algebraic approach to make further progress on the problem. 
After introducing the necessary notation and central axioms in \Cref{sec: prelims}, we reduce the question of checking whether the characterization holds to determining the rank of a matrix by weakening \emph{ex post} efficiency and strategyproofness in \Cref{sec: matrix-interpretation}. Based on this idea, we devise an algorithm that determines the rank of the given matrix.
in \Cref{sec: algo}. In \Cref{sec:dep}, we prove  that the set of strategyproof and \emph{ex post} efficient assignment rules does not have the DEP property for $n \ge 3$, and that the \rsd characterization does not hold in a restricted domain of preferences introduced by \citet{ChCh16a}.  Finally, the paper concludes in \Cref{sec: results}.    

\section{Preliminaries}\label{sec: prelims}

Let $N$ be a set of agents and $H$ a set of houses with $|N|=|H|=n$. A \emph{preference profile} $R$ associates with each agent $i\in N$ a preference ordering $\succ_i$ over the houses. The set of all preference profiles is denoted by \domain. \emph{Random assignments} are represented by \emph{bistochastic matrices} $(M_{i,h})_{i\in N, h\in H}$ where $M_{i,h} \ge 0$ and $\sum_{h' \in H}M_{i,h'}=\sum_{i' \in N}M_{i',h}=1$ for all $i \in N$ and $h \in H$. The \emph{support} of a random assignment $M$ is the set of agent-house pairs $(i,h)$ for which $M_{i,h}>0$.
Whenever $M_{i,h}\in\{0,1\}$ for all agent-house pairs $(i, h)$, $M$ is a permutation matrix and represents a \emph{deterministic assignment}.

A probabilistic assignment rule $f$ maps each profile $R$ to a bistochastic matrix $f(R)$ where, with slight abuse of notation, the entry $f(R,i,h)$ in the $i$th row and $h$th column of the matrix corresponds to the probability of agent $i$ receiving house $h$ in profile $R$ and $f(R,i)$ denotes agent $i$'s assignment in $R$.

In the following, we formally define \rsd and the axioms required for the characterization.

\begin{definition}
	Given a profile $R \in \domain$, a deterministic assignment $M$ is (Pareto) \emph{efficient} if there exists no deterministic assignment $M'\neq M$ such that for all $i\in N$ and $h,h'\in H$ with $h \neq h'$, $M'_{i,h'}=M_{i,h}=1$ implies $h'\succ_i h$.
	An assignment rule is \emph{ex post efficient} if for all $R \in \domain$, $f(R)$ can be represented as a convex combination of efficient deterministic assignments. 
\end{definition}

Let $\Pi$ be the set of all (priority) orders over the agents. Denote serial dictatorship for a specific priority order $\pi \in \Pi$ by $\serd_\pi$. For a given profile $R$, each deterministic efficient assignment coincides with the outcome of a serial dictatorship on $R$ \citep[see, e.g.,][]{Mane07a}. Therefore, an assignment rule satisfies \emph{ex post} efficiency if for all $R \in \domain$, there exist weights $\lambda_{\pi}^R \ge 0$ with $\sum_{\pi \in \Pi} \lambda_{\pi}^R=1$ such that $f(R)=\sum_{\pi \in \Pi}\lambda_{\pi}^R \serd_\pi(R)$.

\rsd can now be defined by choosing $\lambda_\pi^R=\nicefrac{1}{n!}$ for every $\pi$ and $R$, i.e.,

\begin{align*}
	\rsd(R)=\sum_{\pi \in \Pi}\frac{1}{n!}\serd_\pi(R).
\end{align*}
Furthermore, we say that a rule coincides with \rsd if it returns the same random assignment as \rsd for each profile.

A variant of \emph{ex post} efficiency merely requires that for each profile the support of the resulting random assignment coincides with that of some \emph{ex post} efficient random assignment. In other words, the support has to be a subset of that of \rsd.

\begin{definition}
	An assignment rule f is \emph{support efficient} if for all $R \in \domain$, $i \in N$, and $h \in H$, $f(R,i,h)=0$ whenever $\serd_\pi(R,i,h)=0$ for all $\pi \in \Pi$. Equivalently, $f$ is support efficient if for all $R \in \domain$, $i\in N$, and $h\in H$, $\rsd(R,i,h)=0$ implies $f(R,i,h)=0$.
\end{definition}

Support efficient is weaker than \emph{ex post} efficiency, but the conditions coincide when $n \le 3$.

\begin{proposition}\label{pro:supp-expo}
	Support efficiency and \emph{ex post} efficiency are equivalent for $n \le 3$.
\end{proposition}

\begin{proof}
	The case $n = 2$ is easily solved by exhausting all cases. If the two agents disagree on their top choice, only one deterministic assignment is efficient. Therefore all assignments that violate \emph{ex post} efficiency also violate support efficiency.
	Otherwise, the two agents share the same preferences, in this case all random assignments are \emph{ex post} efficient and thus also support efficient.
	
	For the case $n = 3$, assume that a preference profile $R$ and random assignment $f(R)$ exist such that $f(R)$ is support efficient but not \emph{ex post} efficient.
	Then, there exists a deterministic assignment $M$ that is not efficient but needed to represent $f(R)$.
	Furthermore, by support efficiency, the support of $M$ is efficient.
	
	We consider two cases.
	$M$ can be made efficient either by letting three agents trade their houses in a circular fashion, or by swapping the houses of two agents.

	In the first case, no agent received her top choice in $M$, implying that not all agents rank the same house first. If each house is ranked first by some agents, support efficiency ensures that each agent receives her top choice, so $M$ is efficient. Otherwise, w.l.o.g., agents $1$ and $2$ rank $h_1$ first, whereas agent $3$ ranks $h_2$ first but receives $h_1$ under $M$. However, by support efficiency, it is not permitted that agent $3$ receives $h_1$ in such profiles.   
	
	For the second case, two agents, w.l.o.g. $1$ and $2$, both improve when they swap houses $h_1$ and $h_2$, i.e., $h_1 \succ_1 h_2$ and $h_2 \succ_2 h_1$ but $1$ receives $h_2$ and $2$ receives $h_1$ in $M$.
	Assume now, again w.l.o.g., that $h_1 \succ_3 h_2$.
	It is obvious that in this case agent $2$ cannot receive $h_1$ in any efficient deterministic assignment.
	Again, $M$ violates support efficiency.
	
	We have shown that for $n = 3$, a violation of \emph{ex post} efficiency implies a violation of support efficiency.
	Since \emph{ex post} efficiency implies support efficiency, they are equivalent for $n = 3$.
\end{proof}

We now give an example for $4$ agents in which support efficiency is strictly weaker than \emph{ex post} efficiency.

\begin{example}\label{eg:supp}
	Let the preference relations of agents $1$ and $2$ be $h_1 \succ h_2 \succ h_3 \succ h_4$ and $h_2 \succ h_1 \succ h_3 \succ h_4$ be the preferences of agents $3$ and $4$.
	Consider the random assignment where agents $1$ and $2$ receive the lottery $p(h_1) = 0$, $p(h_2) = \frac{1}{2}$, $p(h_3) = p(h_4) = \frac{1}{4}$ and agents $3$ and $4$ receive the lottery $p(h_1) =  \frac{1}{2}$, $p(h_2) = 0$, $p(h_3) = p(h_4) = \frac{1}{4}$.
	This assignment violates \emph{ex post} efficiency because each efficient deterministic assignment assigns either $h_1$ to agent $1$ or $2$ or it assigns $h_2$ to agent $3$ or $4$.
	Since agents $1$ and $2$ never receive $h_1$ and agents $3$ and $4$ never receive $h_2$ from the random assignment, it cannot be represented as a distribution over efficient deterministic assignments.
	The assignment satisfies support efficiency since each house can go to each agent in some efficient deterministic assignment.
\end{example}

To judge whether an agent $i$ is able to beneficially misreport her preferences, we, analogously to \citet{BoMo01a}, assume that agent $i$ has a von Neumann-Morgenstern utility function $u_i$ which is consistent with $\succ_i$. This means that there exist $u_i\colon H\rightarrow \mathbb{R}$ such that $u_i(f(R))=\sum_{h \in H}u_i(h) f(R,i,h)$, and $u_i(h_k)>u_i(h_l)$ if and only if $h_k \succ_i h_l$. Since the concrete utility function is unknown, a manipulation counts as beneficial if there exists a utility function $u_i$ consistent with $\succ_i$ for which it is beneficial. A rule without such manipulation incentives is called strategyproof.\footnote{This version of strategyproofness for probabilistic assignment rules is sometimes also called (strong) \sd-strategyproofness \citep[see, e.g.,][]{Bran17a}.}

\begin{definition}
	An assignment rule $f$ is \emph{strategyproof} if for all $R,R' \in \domain$ with ${\succ_j}={\succ'_j}$ for all $j \in N \setminus \{i\}$, $\sum_{h' \succ_i h} f(R,i,h') \ge \sum_{h' \succ_i h} f(R',i,h')$ for every $h \in H$.
\end{definition}

To implement strategyproofness, we leverage a result from \citet{Gibb77a}, which shows that a mechanism is strategyproof if and only if it is localized and nonperverse. In particular, it suffices to consider swaps of two houses that are adjacent in the manipulator's ranking.\footnote{\citeauthor{Gibb77a} considers the general social choice domain. \citet{MeSe21a} have rediscovered this equivalence in the context of random assignment.}
 
\begin{definition}
	Let $R,R' \in \domain$, $i\in N$, and $h_k,h_l\in H$ such that ${\succ_j}={\succ'_j}$ for all $j \in N \setminus \{i\}$ and ${\succ'_i}={\succ_i}\setminus\{(h_k,h_l)\}\cup\{(h_l,h_k)\}$.
	An assignment rule $f$ is
	\begin{itemize}
		\item \emph{localized} if $f(R,i,h)=f(R',i,h)$ for all $h\in H\setminus\{h_k,h_l\}$, and
		\item \emph{nonperverse} if $f(R,i,h_k) \ge f(R',i,h_k)$ and $f(R,i,h_l) \le f(R',i,h_l)$.
	\end{itemize}
\end{definition}

It turns out that weakening strategyproofness to localizedness 
eliminates the inequality constraints imposed by nonperverseness.

\begin{definition}
	An assignment rule $f$ satisfies \emph{equal treatment of equals} if for all $R \in \domain$ and $i,j \in N$ with ${\succ_i}={\succ_j}$, $f(R,i,h)=f(R,j,h)$ for all $h \in H$.
\end{definition} 
Thus, equal treatment of equals ensures that agents with the same preferences receive the same assignment.

Finally, we introduce a natural property that is helpful for reducing the number of profiles a rule needs to be defined on.

\begin{definition}
	An assignment rule $f$ is \emph{symmetric} if for all $R \in \domain$, any permutation of the agents $\pi: N \to N$ we have $\pi \circ f(R)=f(\pi \circ R)$ and for any permutation of the houses $\tau: H \to H$ we have $\tau \circ f(R)= f(\tau \circ R)$. Here, $\pi$ permutates the rows and $\tau$ permutates the columns of $R$ and $f(R)$.
\end{definition}

Loosely speaking, a symmetric rule does not take into account the identities of agents and houses.

\begin{remark}
The two conditions of symmetry are known as anonymity and neutrality in the more general domain of social choice \citep[see, e.g.,][]{Zwic15a}. Within the assignment domain, anonymity cannot be considered in isolation because agents are indifferent between assignments in which they receive the same house. Viewing agents as voters and deterministic assignments as alternatives, permutations via neutrality allow for permuting assignments, not houses. 
	Permuting two voters $i$ and $j$ via anonymity results in an ``illegal'' assignment profile because agent $i$ is indifferent between assignments in which agent $j$ receives the same house and vice versa. This can be rectified by permuting assignments accordingly. As a consequence, anonymity should only be considered in conjunction with neutrality in the assignment domain. 
\end{remark}

Note that symmetry is a stronger axiom than equal treatment of equals.
\begin{restatable}{proposition}{symmetry}
	Every symmetric assignment rule satisfies equal treatment of equals.
\end{restatable}
\begin{proof}
	Let $f$ be a symmetric assignment rule and $R$ be an arbitrary profile with ${\succ_i}={\succ_j}$ for two agents $i,j \in N$. Consider the permutation $\pi=(i j)$ that only swaps the identities of agents $i$ and $j$. As ${\succ_i}={\succ_j}$, $R=\pi \circ R$ implies $f(R)=\pi \circ f(R)$ by symmetry. In particular, $f(R,i)=\pi \circ f(R,i)=f(R,j)$ showing that agents $i$ and $j$ receive the same assignment under $f$ in $R$. 
\end{proof}

To see that equal treatment of equals does not imply symmetry, consider $n=2$ and the assignment rule $f$ with $f(R)=\rsd(R)$ for the two profiles where both agents have the same preferences. For the other two profiles $R'$ and $R''$ where both agents have different preferences, let $f(R',1)=(1,0)$ and $f(R'',1)=(1,0)$. Clearly, $f$ satisfies equal treatment of equals. However, moving from $R'$ to $R''$ by permuting the two houses does not permute the assignments. In both profiles, agent $1$ receives $h_1$, contradicting $\tau \circ f(R')=f(\tau \circ R')=f(R'')$. 

Symmetry imposes an equivalence class structure on \domain that allows $f$ to be well-defined by only defining it on the set of canonical profiles $\canonical \subset \domain$ which contains one representative profile for each equivalence class that is chosen according to some predefined order over \domain.
We will show that positive results for \canonical carry over to \domain without imposing symmetry, a necessary simplification step given that $|\domain|=n!^n$.

\section{A linear algebraic view on the problem}\label{sec: matrix-interpretation}

	Our overall goal in this section is to describe the set of all rules that satisfy equal treatment of equals, \emph{ex post} efficiency, and strategyproofness by a system of linear equations. 
	
	To this end, note that, for fixed $n$, all axioms except \emph{ex post} efficiency are defined and can be represented by constraints in terms of a vector $\vect{x}=\left(\vect{x}_{(R,i,h)}\right) \in \mathbb{R}^{n^2 n!^n}$ where $\vect{x}_{(R,i,h)}$ corresponds to $f(R,i,h)$. By contrast, efficiency constraints require us to represent an assignment rule $f$ by a vector $\vect{x}=\left(\vect{x}_{(R,\pi)}\right) \in \mathbb{R}^{n! n!^n}$ where $\vect{x}_{(R,\pi)}$ corresponds to the weight $\lambda_{\pi}^R$ of $\serd_\pi$ in profile $R$. 
	
	Generally, it is possible to also represent the other axioms in terms of $\vect{x}_{(R,\pi)}$, e.g., for equal treatment of equals, one has to find the set of all combinations of serial dictatorships that yield the same probabilistic assignment for both agents and each profile where two agents $i$ and $j$ have the same preferences. This can be achieved by requiring that the sum of the weights of all serial dictatorships where $i$ receives house $h$ has to equal the sum of the weights of all serial dictatorships where $j$ gets $h$.
	However, the representation of $f$ in terms of $\left(\vect{x}_{(R,\pi)}\right)$ is not unique \citep[see, e.g.,][]{PyTr23a} and requires $n! n!^n$ instead of $n^2 n!^n$ variables.
	
Weakening \emph{ex post} efficiency to support efficiency enables the representation of efficiency via $f(R,i,h)$. On top of that, we also weaken strategyproofness to localizedness due to the fact that nonperverseness is the only axiom (apart from the nonnegativity part of the bistochastic matrix constraints) that cannot be written in terms of linear equations.

	\begin{conjecture}\label{con:strong-rsd}
		\rsd is the only assignment rule that satisfies equal treatment of equals, support efficiency, and localizedness.
	\end{conjecture}

	Proving this statement immediately implies that \rsd is characterized by equal treatment of equals, \emph{ex post} efficiency, and strategyproofness.
	In case the statement does not hold, a counterexample might give us new insights and ideas to construct a counterexample for the original characterization. In particular, each counterexample of the original conjecture must also be a counterexample for \Cref{con:strong-rsd}.
	
	We now reformulate the problem as a system of linear equations such that every rule satisfying all axioms from \Cref{con:strong-rsd} is a solution to the system.
	As already mentioned, we can represent assignment rules $f$ as vectors \vect{x}, where $\vect{x}_{(R, i, h)} = f(R, i, h)$ for all profiles $R$, agents $i$, and houses $h$.
	The constraints induced by the axioms are represented by the rows of a matrix \vect{A} and a vector \vect{b}, such that $\vect{Ax} = \vect{b}$ if $f$ represented by \vect{x} satisfies all axioms. The columns of \vect{A} correspond to the triples $(R,i,h)$.
	Define $\vect{e}_{(R,i,h)} \in \mathbb{R}^{1 \times n^2 n!^n}$ as the unit vector with $1$ at entry $(R,i,h)$ and $0$ otherwise. The rows of \vect{A} have the following form depending on the type of axiom. 
	
	\begin{enumerate}
		\item Bistochasticity constraints (excluding nonnegativity constraints): \vect{A} contains a row $\vect{a}_k$ for each profile $R$
		\begin{enumerate}
			\item and agent $i$, with $\vect{a}_{k} = \sum_{h \in H} \vect{e}_{(R,i,h)}$, and
			\item and house $h$, with $\vect{a}_{k} = \sum_{i \in N} \vect{e}_{(R,i,h)}$.
		\end{enumerate}
		For such rows, $\vect{b}_k=1$.
		\item Support efficiency: \vect{A} contains a row $\vect{a}_k$ for each triple $(R, i, h)$ satisfying $\rsd(R,i,h)=0$, with $\vect{a}_{k} = \vect{e}_{(R,i,h)}$. For such rows, $\vect{b}_k=0$.
		\item Localizedness: \vect{A} contains a row $\vect{a}_k$ for each profile $R$, agent $i$, house $h$, and each possible adjacent swap to profile $R'$ that agent $i$ can perform that does not move house $h$, with $\vect{a}_{k} = \vect{e}_{(R,i,h)}-\vect{e}_{(R',i,h)}$. For such rows, $\vect{b}_k=0$.
		\item Equal treatment of equals: \vect{A} contains a row $\vect{a}_k$ for each profile $R$, house $h$, and agent pair $(i, j)$ such that $i \ne j$ and ${\succ_i}={\succ_j}$, with $\vect{a}_{k} = \vect{e}_{(R,i,h)}-\vect{e}_{(R,j,h)}$. For such rows, $\vect{b}_k=0$.
	\end{enumerate}

As \rsd satisfies all axioms, $\vect{A}\vect{x}^{\rsd} = \vect{b}$, where $\vect{x}^\rsd$ is the vector representing \rsd.

In general, it does not hold that every solution to $\vect{Ax} = \vect{b}$ corresponds to a valid assignment rule since the nonnegativity of variables $\vect{x}_{(R,i,h)}$ 
is not guaranteed. Nevertheless, the structure of \rsd allows us to mix any other solution with $\vect{x}^\rsd$ in a way that returns a new assignment rule satisfying all axioms.

\begin{proposition}\label{prop:rank}\label{thm:rank}
	Let $\vect{y}\neq \vect{x}^\rsd$ be a solution to $\vect{Ax} = \vect{b}$. Then, there exists $\lambda>0$ such that $\lambda \vect{y}+(1-\lambda)\vect{x}^\rsd$ is an assignment rule that satisfies all axioms and differs from \rsd.
\end{proposition}

\begin{proof}
	Apart from nonnegativity, $\lambda \vect{y}+(1-\lambda)\vect{x}^\rsd$ satisfies all axioms for all $\lambda \in [0,1]$ as
	\begin{align*}
		\vect{A}(\lambda \vect{y}+(1-\lambda)\vect{x}^\rsd)=\lambda \vect{Ay}+(1-\lambda)\vect{A}\vect{x}^\rsd=\vect{b}.
	\end{align*}
	In order to ensure nonnegativity, choose $\lambda^*>0$ such that $\lambda^* \vect{y}_{(R,i,h)} +(1-\lambda^*)\vect{x}_{(R,i,h)}^\rsd \ge 0$ for all $(R,i,h)$. This is possible due to the fact that $\vect{x}_{(R,i,h)}^\rsd =0$ implies $\vect{y}_{(R,i,h)}=0$ as \vect{y} satisfies support efficiency.
	
Thus, $\lambda^* \vect{y}+(1-\lambda^*)\vect{x}^\rsd$ corresponds to an assignment rule that satisfies all axioms and differs from \rsd since the representation of a rule in terms of $(\vect{x}_{(R,i,h)})$ is unique by definition. 
\end{proof}

\Cref{prop:rank} shows that whenever there exists a solution $\vect{y} \neq \vect{x}^\rsd$ to $\vect{Ax} = \vect{b}$, \Cref{con:strong-rsd} cannot hold. Furthermore, $\vect{y}-\vect{x}^\rsd \neq \vect{0}$ lies in the kernel $\ker(\vect{A})$ of \vect{A}.  

\begin{corollary}\label{cor:eq}
	The following statements are equivalent:
	\begin{itemize}
		\item \Cref{con:strong-rsd} holds, i.e., the only solution to $\vect{Ax} = \vect{b}$ is $\vect{x}^\rsd$.
		\item \vect{A} has full rank, i.e., $\mathit{rank}(\vect{A}) = n^2n!^n$.
		\item $\ker(\vect{A}) = \{\vect{0}\}$.
	\end{itemize}
\end{corollary}

In the next section, we use \Cref{prop:rank} and these equivalences to devise an algorithm.
We believe that \Cref{cor:eq} could also be helpful for finding a general analytic proof of \Cref{con:strong-rsd}.

\section{Checking whether the matrix has full rank}\label{sec: algo}
	
The following algorithm shows that \Cref{con:strong-rsd} holds 
by proving that \vect{A} has full rank.
\Cref{thm:rank} shows that this is equivalent to proving \Cref{con:strong-rsd} which in turn implies the original \rsd characterization.
In principle, the rank of \vect{A} can be computed using standard methods such as Gaussian elimination.
However, there are two main issues with that approach.
First, the size of the matrix is larger than $n^2n!^n \times n^2n!^n$, see \Cref{fig:columnnumbers} for explicit numbers.

\begin{figure}[bp]\centering
        
		\begin{tabular}{c|cccccc}
            $n$ & 2 & 3 & 4 & 5 & 6 & 7 \\\midrule
            $n^2n!^n$ & 16 & 1944 & $5.3 \cdot 10^6$ & $6.2 \cdot 10^{11}$ & $5 \cdot 10^{18}$ & $4 \cdot 10^{27}$
        \end{tabular}
		
		\caption{Number of columns of $\vect{A}$ depending on $n$.}
		\label{fig:columnnumbers}
\end{figure}
		
This can be partially mitigated because the matrix is sparse.
Even though most entries are zero and do not need to be stored in memory, the remaining matrix is still very large.
Second, standard methods often run into numerical problems.
	
	To circumvent these issues, the algorithm we propose in this section uses search to construct all $n^2n!^n$ rows $\vect{e}_{(R, i, h)}$ using elementary row operations implying that the matrix has full rank. 
	In particular, we add or substract multiples of one row from another or multiply a row by $-1$.
	Division is only used when we found a row that has only one non-zero entry to normalize.
	In this way the algorithm is guaranteed to not run into any numerical problems.
	Furthermore, we never explicitly construct the matrix, and use the symmetry of the domain to simplify the computation.

	The main idea of the algorithm builds on the fact that localizedness is the only axiom which connects profiles, i.e., the rows of matrix \vect{A} have nonzero entries in different preference profiles.
	On the contrary, for all other axioms, the rows have nonzero entries only for a single profile.
	Starting with some preference profile $R_s$ where all agents share the same preferences, it is possible to build the rows $\vect{e}_{(R_s, i, h)}$ for all agent-house pairs $(i, h)$ using elementary row operations.
	This can be done by adding ``equal treatment of equals rows'' to ``bistochasticity rows'' until the only nonzero entry is at index $(R, i, h)$.
	With this method we can construct $\vect{e}_{(R_s,i,h)}$ for all agent-house pairs $(i,h)$. From an axiomatic point of view, it is clear that all agents need to receive the same assignment in $R_s$.
	
	Next, the new rows $\vect{e}_{(R_s, i, h)}$ can be added to the localizedness rows to build new rows $\vect{e}_{(R', i, h)}$ for profiles $R'$ that can be reached by swap manipulations from $R_s$.
	The algorithm can then try to solve these profiles, find new rows, and then propagate them further.
	
	Thus, the algorithm consists of two parts, namely
\begin{itemize}
	\item a subroutine that evaluates a single profile $R$ and builds as many rows $\vect{e}_{(R, i, h)}$ using elementary row operations as possible, and
	\item the main loop which builds rows for profiles that can be reached using localizedness and chooses the next profile to evaluate.
\end{itemize}
In contrast to $R_s$, note that in general, it is not possible to completely ``solve'' a profile at the first visit. 
	Therefore, the main loop uses a priority queue to track which profile received the most rows since it was last considered.
	Guiding the search using this heuristic improves the runtime of the algorithm over naive breath first search or depth first search.
	
	The algorithm continues the search until the identity matrix is contained in \vect{A} or it proves that this is not possible.
	For that, it keeps track of the triples $(R, i, h)$ for which the row $\vect{e}_{(R, i, h)}$ was constructed with an indicator function $\isrsd: \domain \times N \times H \rightarrow \{1, 0\}$ that returns $1$ if the row $\vect{e}_{(R, i, h)}$ is already contained in the matrix and $0$ otherwise.
	This indicator function is updated during program execution.
	When we refer to $\isrsd$, we refer to the current state of algorithm execution, unless stated otherwise.
	At the start of the algorithm $\isrsd \equiv 0$ is initialized to be $0$ for every triple.
	In a first step, it sets $\isrsd(R, i, h) = 1$ for all triples $(R, i, h)$ with $\rsd(R,i,h)=0$ since for those, $\vect{a}_k = \vect{e}_{(R, i, h)}$ by definition.
	Once $\isrsd \equiv 1$, the algorithm terminates as it has shown that the matrix \vect{A} has full rank.
	
	We first present the subroutine, then the complete algorithm.

	\subsection{Solving single profiles}
	
	Given a preference profile $R$ and indicator function $\isrsd$, the following subroutine computes all agent-house pairs $(i, h)$ for which the vector $\vect{e}_{(R, i, h)}$ can be constructed.
	We start by writing the rows corresponding to the bistochasticity and equal treatment of equals constraints of $R$ into a separate matrix $\vect{B}$.
	The main idea then is to simplify these rows by setting all indices $(R, i, h)$ to zero if $\isrsd(R, i, h) = 1$. In this way, we incorporate constraints from support efficiency and localizedness into $\vect{B}$.
	This is allowed since $\isrsd(R, i, h) = 1$ implies $\vect{e}_{(R, i, h)}$ was constructed which in turn allows us to add or substract it from each row in $\vect{B}$ such that the entry becomes $0$.
	
	If the resulting matrix \vect{B} contains rows with only one nonzero entry at position $(R, i, h)$, then we set $\isrsd(R, i, h)$ to $1$ and go back to the previous step.
	Otherwise, no simplifications are possible.
	We check if combining the resulting equal treatment of equals and bistochasticity rows results in new rows with only one nonzero entry.
	To do so, it is sufficient to check for each bistochasticity row $\vect{b}$ and house $h$ if for all agents $i$ with $b_{(i, h)} = 1$ we have for all agents $j$ that $b_{(i, h)} = b_{(j, h)} = 1$ if and ony if $\succ_i = \succ_j$.
	If this is the case, we can construct the rows $\vect{e}_{(R, i, h)}$ for all $i$ for which $b_{(i, h)} = 1$ by adding the equal treatment of equals rows to the bistochasticity row.
	For these rows, the algorithm can once again set $\isrsd(R, i, h) = 1$ and go back to the first step.
	Otherwise, if no new rows are found, the subroutine terminates and returns the updated indicator function.
	
	The subroutine only uses elementary row operations to construct new rows.
	Furthermore, it can restrict the matrix to a single profile $R$ since it only considers matrix rows that have only zero entries for all indices of other profiles.
	Thus, these operations do not alter the rank of the matrix \vect{A}.

	Another important property of the subroutine is that it is symmetric with respect to inputs.
	In other words, if we permute all inputs with some permutation of the agents $\pi \in \Pi$ and houses $\tau \in \Tau$, the updates to the indicator function are permuted by the same permutation.
	This property follows from the fact that the algorithm is deterministic and permutations of the profile permute the indices of the matrix $\vect{B}$ in the same way.
	Thus, the results are the same up to permutation.

	\begin{algorithm}[tbp]
		\caption{Subroutine that constructs new rows $\vect{e}_{(R, i, h)}$ for input profile $R$.}
		\label{alg:rsd-subroutine}
		\begin{algorithmic}
			\Input
			\Desc{$R$}{Preference profile}
			\Desc{$\isrsd$}{Function $\isrsd: \domain \times N \times H \rightarrow \{0, 1\}$}
			\EndInput
		\end{algorithmic}
		\begin{algorithmic}[1]
			\State $\vect{B} \gets $ Matrix with bistochasticity and equal treatment of equals rows for profile $R$
			\While{\isrsd was updated}
				\While{\isrsd was updated}
					\ForAll{$(i, h) \in N \times H$}
						\If{$\rsd(R, i, h) = 1$}
							\State $ b_{(i, h)} \gets 0$ for all rows $\vect{b}$ in $\vect{B}$
						\EndIf
					\EndFor
					\ForAll{Rows $\vect{b}$ in $\vect{B}$}
						\If{$\exists (i, h) \in N \times H$ such that $\vect{b} = \vect{e}_{(R, i, h)}$ and $\isrsd(R, i, h) = 0$}
							\State $\isrsd(R, i, h) \gets 1$.
						\EndIf
					\EndFor
				\EndWhile
				\ForAll{Rows $\vect{b}$ in $\vect{B}$ and $h \in H$}
					\If{$\forall i \in N: b_{(i, h)} = 1 \Rightarrow \forall j \in N\, (b_{(j, h)} = 1 \Leftrightarrow  {\succ_i} = {\succ_j})$}
						\ForAll{$i \in N$ \textbf{if} $ b_{(i, h)} = 1$}
							\State $\isrsd(R, i, h) \gets 1$.
						\EndFor
					\EndIf
				\EndFor
			\EndWhile
		\end{algorithmic}
	\end{algorithm}

	\subsection{Guided search and localizedness}

	\Cref{alg:rsd-subroutine} is able to evaluate single profiles.
	To complete the algorithm, we still need to decide the evaluation order of the profiles and combine the new rows with localizedness.
	The full algorithm is described in \Cref{alg:rsd}.
	
	The first step is to initialize the indicator function \isrsd that keeps track of the rows $\vect{e}_{(R, i, h)}$ that where already build.
	We initialize all entries with $0$, before setting all entries $(R,i,h)$ with $\rsd(R,i,h)=0$ to $1$ as this implies $f(R,i,h)=0$ by support efficiency.
	
	Then, we use a standard best-first search algorithm to choose which profile to evaluate next.
	The heuristic used to determine the priority of profile $R$ is the number of rows $\vect{e}_{(R, i, h)}$ that where constructed since the last time the profile was considered.
	The priority queue is initialized with the profile $R_s$ where all agents have the same preferences.
	This profile is a good choice since the submatrix of this profile has full rank and the bistochasticity and equal treatment of equals constraints are already sufficient to construct all $\vect{e}_{(R_s, i, h)}$.
	Although we use a search algorithm, it has no ``goal profile'' in the usual sense but rather searches until it completed the indicator function or fails to do so.
	The advantage of best-first over depth-first or breath-first search is that it is much faster as it first evaluates profiles that are likely to be solved completely by the subroutine \Cref{alg:rsd-subroutine}.
	We observed that other methods visit the same profiles more frequently on average.
	
	The algorithm then combines the rows found by the subroutine with localizedness by multiplying the localizedness row with $-1$ if necessary and adding the row from the subroutine.
	More precisely, if $\isrsd(R, i, h) = 1$ and agent $i$ manipulates by rearranging houses above and below $h$, then $\isrsd(R', i, h) \gets 1$, where $R'$ is the profile agent $i$ manipulates to.
	Therefore, we can set $\isrsd(R', i, h) = 1$ if $\isrsd(R, i, h) = 1$.
	We further reduce the number of manipulations that need to be considered by only allowing swap manipulations of adjacent houses.
	However, this does not really constitute a restriction since the same manipulations can be carried out by performing multiple swaps, i.e., all other manipulations are linearly dependent on ``pairwise swap'' rows.
	
	This algorithm is still not efficient enough to solve the case of $n = 5$.
	In order to reduce the size of \domain, we take advantage of the symmetry of the axioms and prove that the algorithm can assume symmetry without loss of generality.
	In particular, we show that the result of the algorithm on all canonical profiles \canonical generalizes to \domain when ensuring that a manipulation that leaves the domain falls back to a canonical profile.
	For example, if agent $i$ manipulates from profile $R \in \canonical$ to $R' \in \domain \setminus \canonical$ then the algorithm assumes $i$ manipulated from $R$ to $\mathit{canonical}(R', i)$, where $\mathit{canonical}$ is a function that maps a profile to the canonical profile.
	
	A very important detail here is that while this function always maps to a single profile, the manipulating agent $i$ might map to multiple agents in the canonical profile.
	To account for this we let the function $\mathit{canonical}$ also return a list of agents in the new profile that the manipulator can map to.
	If we want to use \Cref{alg:rsd} on \domain instead of \canonical, $\mathit{canonical}$ becomes the identity function and simply returns the corresponding profile and agent.
	
		\begin{algorithm}[tbp]
		\caption{Verify \rsd Characterization}
		\label{alg:rsd}
		\begin{algorithmic}
			\Input \Desc{$n$}{Number of Agents and Objects}
			\EndInput
		\end{algorithmic}
		\begin{algorithmic}[1]
			\State $\isrsd \gets 0$ \Comment{Initialize $\isrsd: \canonical \times N \times U \rightarrow \{1, 0\}$ as the constant $0$ function.}
			\ForAll{$(R, i, h) \in \domain \times N \times H$}
			\If{$\rsd(R, i, h) = 0$}\Comment{$f(R,i,h)=0$ due to support efficiency.}
			\State $\isrsd(R, i, h) \gets 1$
			\EndIf
			\EndFor
			
			\State \emph{queue} $\gets$ \textbf{new} \emph{Priority Queue}
			\State $\mathit{queue.insert}(R_s, 0)$
			\While{\emph{queue} \textbf{is not} \emph{empty}}
			\State $R \gets \mathit{queue.findmax}()$
			\State $\mathit{queue.deletemax}()$
			\State\label{step:QP}$\mathit{\Cref{alg:rsd-subroutine}}(R, \isrsd)$ \Comment{\Cref{alg:rsd-subroutine} updates \isrsd.}
			\ForAll{$R'$ s.t. $\exists i \in N\ \forall j \ne i\ \succ_j = \succ'_j \land \exists k \in [n]\ \succ'_i = \mathit{swap}(\succ_i, k, k+1)$}\label{step:neighbor}
			\State $R^*, \mathit{manipulators} = \mathit{canonical}(R', i)$
			\State $\Delta \gets 0$
			\ForAll{$l \in [n] \setminus \{k, k + 1\}$}
			\ForAll{$i^* \in \mathit{manipulators}$}
			\State $h \gets \mathit{lth\ best}(\succ_i, l)$
			\State $h^* \gets \mathit{lth\ best}(\succ^*_{i^*}, l)$
			\If{$\isrsd(R, i, h)=1$ \textbf{and} $\isrsd(R^*, i^*, h^*)\neq 0$}
			\State $\isrsd(R^*, i^*, h^*) \gets 1$
			\State $\Delta \gets \Delta + 1$
			\EndIf
			\EndFor
			\EndFor
			\If{$\Delta > 0$}
			\If{$R' \in \mathit{queue}$}
			\State $\mathit{queue.increasepriority}(R^*, \Delta)$
			\Else
			\State $\mathit{queue.insert}(R^*, \Delta)$
			\EndIf
			\EndIf
			\EndFor
			\EndWhile
			\State \Return $\isrsd \equiv 1$ \Comment{The characterization holds if $\isrsd$ equals $1$ for every $(R,i,h)$.}
		\end{algorithmic}
	\end{algorithm}
	
	\begin{lemma}\label{lem:sym}
		The result of \Cref{alg:rsd} holds for \domain when the search space is restricted to \canonical.
	\end{lemma}
	
	\begin{proof}
		We show that \Cref{alg:rsd} on \canonical is equivalent to \Cref{alg:rsd} on \domain by induction.
		Let $\isrsd: \domain \times N \times H \rightarrow \{0, 1\}$ and $\isrsd^*: \canonical \times N \times H \rightarrow \{0, 1\}$ be the indicator functions for the first and second program, respectively.
		Denote $\Pi$ as the set of all permutations of agents and $\Tau$ as the set of all permutations of the houses, i.e., $\pi \in \Pi$ and $\tau \in \Tau$ map a preference profile $R$ to another preference profile $R'=\pi(\tau(R))$ by rearranging the agents according to a permutation $\pi$ and renaming the houses according to a permutation $\tau$.
		Obviously, $|\Pi| = |\Tau| = n!$ as both sets consists of $n!$ permutations of the agents and houses, respectively.
		Our induction proof is based on the idea that \Cref{alg:rsd} on \domain will, after some extra steps, return to a state that is equivalent to \Cref{alg:rsd} on \canonical.
		We show this by induction over the outermost loop of \Cref{alg:rsd}.
		In particular, we show that there exists an execution of \Cref{alg:rsd} on \domain such that the following invariance holds at some point.
		\begin{equation}\label{eq:invariance}
			\isrsd^*(R, i, h) = \isrsd(\pi(\tau(R)), \pi(i), \tau(h)) \quad \forall R \in \canonical, \pi \in \Pi, \tau \in \Tau, i \in N, h \in H
		\end{equation}
		
		\emph{Induction base:} At the start of the algorithm, $\isrsd = \isrsd^* \equiv 0$ meaning that the induction hypothesis trivially holds.
		It still holds after the support efficiency constraints are added to \isrsd since \rsd satisfies symmetry.
		
		\emph{Induction hypothesis:} \Cref{eq:invariance} holds at the start of the $k$-th iteration of the outermost loop.
		
		\emph{Induction step:} We show \Cref{eq:invariance} holds at the end of the $k$-th iteration of the outermost loop.
		\Cref{alg:rsd} will look at profile $R \in \canonical$ in the $k$-th iteration.
		Let the variant on \domain look at all profiles in $[R]$ which denotes the equivalence class of all profiles equivalent to $R$ by symmetry.
		Clearly, both algorithms do not change the indicator value of any profile that is not in $[R]$ or a neighbor of it.
		In line \ref{step:QP}, the algorithm calls the subroutine.
		
		The subroutine \Cref{alg:rsd-subroutine} is deterministic and permutations of the inputs result in the same permutations of the outputs implying that since the second program permutes the inputs, the outputs are also permuted.
		If the second program sets $\isrsd^*(R, i, h) = 1$, the first program is able to set $\isrsd(\pi(\tau(R)), \pi(i), \tau(h)) = 1$ for every $\pi \in \Pi,\tau \in \Tau$ by induction hypothesis.
		Therefore, the invariance condition is preserved for profiles in $[R]$.
		
		Next, in line \ref{step:neighbor}, the algorithm starts to iterate over neighbors of $R$ that can be reached by adjacent swap manipulations of the agents.
		Let $R'$ be an arbitrary neighboring profile, $i$ the manipulating agent, and $k \in [n-1]$ the position in agent $i$'s preferences such that for all $j \ne i$, the preferences stay the same (${\succ_j} = {\succ'_j}$) and ${\succ'_i} = \mathit{swap}(\succ_i, k, k + 1)$.
		Furthermore, let $R'' = \mathit{canonical}(R')$ be the canonical representation of $R'$ and $\pi' \in \Pi$, $\tau' \in \Tau$ be any pair of permutations that maps $R''$ to $R'$.
		For each $l \in [n] \setminus \{k, k + 1\}$ the algorithm performs the following operations.
		Let $h$ be agent $i$'s $l$th most preferred house.
		Then, if $\isrsd^*(R, i, h) = 1$ and $\isrsd^*(R'', i, h) = 0$, set $\isrsd^*(R'', i, h) \gets 1$.
		This is allowed since the localizedness row together with $\vect{e}_{(R, i, h)}$ and multiplication by $-1$ if necessary can reach $\vect{e}_{(R', i, h)}$.
		The first program performs the same operation not only for $R$ but for each profile in $[R]$.
		By induction hypothesis, $\isrsd^*(R, i, h) = \isrsd(\pi(\tau(R)), \pi(i), \tau(h))$ and $\isrsd^*(R'', i, h) = \isrsd(\pi(\tau(R'')), \pi(i), \tau(h))$ for all permutations $\pi \in \Pi$ and $\tau \in \Tau$.
		Thus, the condition of the if statement $\isrsd^*(R, i, h) = 1$ and $\isrsd^*(R'', i, h) = 0$ is true in the second program if and only if it is true in the first program for each permutation $\pi,\tau$.
		Consequently, $\isrsd^*(R'', i, h) = \isrsd(\pi(\tau(R'')), \pi(i), \tau(h)) \gets 1$ for all permutations $\pi$ and $\tau$.
		Again the induction hypothesis is preserved.
		Since no other operations change the indicator function, we conclude that the invariance holds after each step of \Cref{alg:rsd}.
	\end{proof}

	To summarize, it is sufficient to restrict the algorithm to \canonical and all actions of the algorithm can be represented as elementary row operations.
	As they do not change the rank of a matrix and the algorithm shows that the full identity matrix can be constructed from the matrix \vect{A}, we conclude that \vect{A} has full rank.
	\Cref{cor:eq} then implies that \Cref{con:strong-rsd} holds.

	\begin{theorem}
			\rsd is the only assignment rule that satisfies equal treatment of equals, support efficiency, and localizedness when $n\leq 5$.
	\end{theorem}

	\section{Further Results}
	
	In this section, we give counterexamples showing that certain approaches to prove \Cref{con:strong-rsd} (and the weaker original conjecture) are futile.
	
	\subsection{Characterizing \rsd via deterministic extreme points}\label{sec:dep}

	A natural way to achieve fairness is symmetrization.
	The idea is to take a deterministic assignment rule, apply it to every permutation of the agents' roles, and then randomize uniformly over those deterministic assignments. %
	
	We say that a set of assignment rules satisfies the \emph{deterministic extreme point (DEP) property} if it is convex and all of its extreme points are deterministic \citep[see, e.g.,][]{PyUn15a,GPS17a,RoSa20a}.
In this case, the set of assignment rules forms a polyhedron with deterministic extreme points.
A natural way to define sets of assignment rules that can potentially satisfy the DEP property is to consider a set of axioms that can be represented as linear inequalities.
In general, the question whether two sets of assignment rules that are defined by overlapping but different sets of axioms satisfy the DEP property is logically independent.
An important consequence of the DEP property is that all assignment rules in the set can be represented as convex combinations of the deterministic assignment rules that form the finite set of extreme points.

	\citet{Bade13b} conjectured that we can use this to prove a characterization of \rsd by showing that the set of strategyproofness and \emph{ex post} efficient rules satisfies the DEP property.
	\citet{PyUn15a} have shown that the set of strategyproof rules does not satisfy the DEP property by providing a counterexample.
	Unfortunately, this gives no indication as to whether this is also the case when considering strategyproofness in conjunction with \emph{ex post} efficiency.
	In fact, the random dictatorship theorem by \citet{Gibb77a} shows that, in the domain of voting, the set of strategyproof and \emph{ex post} rules satisfies the DEP property, even though the set of all strategyproof rules does not.

	We use a simple equivalence result from linear optimization to find a counterexample for $n = 3$ in the assignment domain.
	We then extend this counterexample to arbitrary $n \ge 3$ and show that the set of strategyproof and \emph{ex post} efficient assignment rules violates the DEP property.
	
	\begin{theorem}
		The set of strategyproof and \emph{ex post} efficient assignment rules violates the DEP property for $n \ge 3$.
	\end{theorem}

	\begin{proof}
		We first show a counterexample for $n = 3$ and then extend it to arbitrary $n > 3$.
		
		Let $P_n$ be a polyhedron defined by the localizedness, support efficiency, row and column sum equalities as well as the nonnegativity, and nonperverseness inequalities for $n$ agents and houses.
		Since $\vect{x}_\rsd \in P_n$, it is nonempty. Furthermore, for any order $\pi$ of the agents, $\serd_{\pi} \in P_n$. Moreover, \Cref{pro:supp-expo} implies that all $\vect{x} \in P_3$ satisfy \emph{ex post} efficiency.
		
		For arbitrary $n$, let $\vect{x} \in P_n$. 
		We say that a linear equality or inequality constraint is active for $\vect{x}$ if it is satisfied with equality.
		The vector $\vect{x}$ is a \emph{basic feasible solution (BFS)} if it satisfies all constraints and there are $n^2n!^n$ linearly independent active constraints. Note that equality constraints are always active. 
		One can then show that $\vect{x}$ is a BFS if and only if it is a extreme point of $P_n$ \citep[see, e.g.,][]{BeTs97a}.

		To show the statement for $n=3$, we therefore have to find a BFS and thus, an extreme point that does not correspond to a deterministic rule. 
		The number of BFSs is in general exponential and naive search might never find a counterexample even if they exist.
		However, for this particular case, we found multiple counterexamples by optimizing in random directions. 
		In this way we found a BFS which does not correspond to a deterministic rule.
		We can conclude that the set of strategyproof and \emph{ex post} efficient rules does not satisfy the DEP property when $n = 3$. The counterexample is specified in \Cref{app:sp-ef-rule}.
		
		To prove the statement for arbitrary $n > 3$, we construct a function $f$ for $n$ alternatives based on a counterexample $g$ for $m<n$ alternatives.
		In particular, we already found a counterexample $g$ for the case $m = 3$.
		We will show that $f$ cannot be represented as a convex combination over deterministic rules in $P_n$ if $g$ cannot be represented as one in $P_m$ and thus $f$ is a valid counterexample for $n$ if $g$ is a counterexample for $m$.
		
		Let $f$ be the function that performs serial dictatorship with the identity order on the first $n - m$ agents and then executes $g$ on the remaining agents and houses. More formally, with priority order $\pi=[n]$, $f(R,i)=\serd_\pi(R,i)$ for each $i \le n-m$ and profile $R \in \mathcal{R}$. Denote by $H^R_m$ the set of houses that were not assigned to one of the first $n-m$ agents. For all other agents $i>n-m$, $f(R,i)=g(R|_{H^R_m},i)$ where $R|_{H^R_m}$ restricts the profile to the preferences of the last $m$ agents over the houses in $H^R_m$. In particular, if of one these last $m$ agents swaps two adjacent houses where at most one of these houses is in $H^R_m$, the assignment returned by $f$ does not change.

		Clearly, $f$ is \emph{ex post} efficient.
		This follows from the fact that the first $n - m$ agents are assigned houses by a serial dictatorship and the last $m$ agents by a rule that is itself \emph{ex post} efficient by assumption.
		For each profile $R \in \domain$, we first decompose the outcome of $g(R)$ and then connect each of those serial dictatorships with the order that $f$ uses for the first $n-m$ agents.
		Thus, we found a decomposition of $f(R)$ as a mixture of serial dictatorships.
		Since $f$ can be represented as a mixture of serial dictatorships for each profile, it is \emph{ex post} efficient.
		
		Next, we show that $f$ is strategyproof.
		The last $m$ agents have no successful manipulation since $g$ is strategyproof and they cannot change the houses shared between them by changing their preferences.
		The first $n - m$ agents also have no successful manipulation since they receive their houses according to a serial dictatorship.
		
		Furthermore, the outcome at each profile is a bistochastic matrix and all entries are nonnegative.
		Therefore, $f$ satisfies all constraints and thus, $f \in P_n$.
		
		Assume now for contradiction that $f$ can be represented as a convex combination over deterministic rules in $P_n$, i.e., there exist deterministic rules $f_1,\dots,f_k \in P_n$ and weights $\lambda_1,\dots, \lambda_k>0$ with $\sum_{\ell=1}^k \lambda_\ell=1$ such that $f=\sum_{\ell=1}^k \lambda_\ell f_\ell$. Consider the subdomain of preferences $\mathcal{R}'$ where the first $n-m$ agents have the same preferences, w.l.o.g. $h_1 \succ h_2 \succ \dots \succ h_n$ and all agents rank $h_1\succ h_2 \succ \dots \succ h_{n-m} \succ h_i$ for any $i>n-m$. By definition of $f$, $f(R,i)=f_\ell(R,i)=h_i$ for $i \le n-m$ and $f(R,i)=g(R|_{H^R_m},i)$ for $i>n-m$. Thus, $g(R|_{H^R_m},i)=\sum_{\ell=1}^k \lambda_\ell f_\ell(R,i)$ for $i>n-m$. Defining $f'_\ell$ as the restriction of $f_\ell$ to $\mathcal{R}'$ for each $\ell \in [k]$, $g(R)=\sum_{\ell=1}^k \lambda_\ell f'_\ell(R)$ for $R \in \mathcal{R}'$. As for each $\ell$, $f'_\ell$ inherits \emph{ex post} efficiency and strategyproofness from $f_\ell$ (which are \emph{ex post} efficient and strategyproof on the larger domain $\mathcal{R}$ by assumption), this contradicts our assumption that $g$ cannot be represented by a convex combination of deterministic rules satisfying \emph{ex post} efficiency and strategyproofness.
		All in all, $f$ cannot be represented as a convex combination over deterministic rules in $P_n$. 
	\end{proof}

	We have leveraged the same approach to construct a counterexample for the set of all assignment rules that satisfy localizedness and support efficiency (or equivalently, \emph{ex post} efficiency) when $n=3$.
	Moreover, we explored whether \emph{non-bossiness}, a property used to characterize serial dictatorships \citep{Sven99a}, might be helpful for characterizing \rsd.
	Non-bossiness requires that if an agent modifies his preferences and still receives the same house, then the allocation of all other agents has to remain the same \citep{SaSo81a}.
	Generalizations of non-bossiness to probabilistic assignment rules \citep[e.g.,][]{Bade16a} fail to be convex.
	To see this, consider two non-bossy rules, a profile in which an agent changes his preferences, and his distribution changes for both rules. When randomizing between both rules, the changes can cancel each other out, but there is no guarantee that the changes for the other agents will also cancel out.
	As a consequence, the set of non-bossy rules does not satisfy the DEP property unless one restricts it further to make it convex.

	\subsection{Characterizing \rsd in subdomains} \label{sec:rsdinsubdomains}
		
		\begin{figure}[tbp]
		\renewcommand{\arraystretch}{1.2}
		\centering
		$
		\begin{array}{cccccccccc}
			&&&&&&h_1 &h_2 &h_3 &h_4 \\\cmidrule{7-10}
			1: &h_1 &h_2 &h_3 &h_4&&\cellcolor{grey} \frac{3}{4} &  0  & \cellcolor{grey} \frac{1}{24}  &  \frac{5}{24} \\
			2: &h_2 &h_1 &h_3 &h_4&&\frac{1}{4}  &  \frac{1}{3}  &  \frac{1}{6}  &  \frac{1}{4} \\
			3: &h_2 &h_3 &h_1 &h_4&	&\cellcolor{grey} 0  & \frac{1}{3}  & \cellcolor{grey} \frac{5}{12}  &  \frac{1}{4} \\
			4: &h_2 &h_3 &h_4 &h_1&&0  &  \frac{1}{3}  &  \frac{3}{8}  &  \frac{7}{24}
		\end{array}
		$\\[0.5em]
		$
		\begin{array}{cccccccccc}
			&&&&&&h_1 &h_2 &h_3 &h_4 \\\cmidrule{7-10}
			1: &h_1 &h_2 &h_3 &h_4 &&\frac{3}{4} &  0  & \cellcolor{grey} 0 & \cellcolor{grey} \frac{1}{4}   \\
			2: &h_2 &h_1 &h_3 &h_4 &&\frac{1}{4}  &  \frac{1}{3}  &  \frac{1}{6}  &  \frac{1}{4}  \\
			3: &h_2 &h_3 &h_4 &h_1 &&0  &  \frac{1}{3}  & \cellcolor{grey} \frac{5}{12}  & \cellcolor{grey} \frac{1}{4}  \\
			4: &h_2 &h_3 &h_4 &h_1 &&0  &  \frac{1}{3}  & \cellcolor{grey} \frac{5}{12}  & \cellcolor{grey} \frac{1}{4}
		\end{array}
		$\\[0.5em]
		$
		\begin{array}{cccccccccc}
			&&&&&&h_1 &h_2 &h_3 &h_4\\\cmidrule{7-10}
			1: &h_2 &h_1 &h_3 &h_4 &&\cellcolor{grey} \frac{1}{2} &  \frac{1}{4}  & \cellcolor{grey} \frac{1}{24}  &  \frac{5}{24}  \\
			2: &h_2 &h_1 &h_3 &h_4 &&\cellcolor{grey} \frac{1}{2}  &  \frac{1}{4}  & \cellcolor{grey} \frac{1}{24}  &  \frac{5}{24} \\
			3: &h_2 &h_3 &h_1 &h_4  &&\cellcolor{grey} 0  &  \frac{1}{4}  & \cellcolor{grey} \frac{1}{2}  &  \frac{1}{4}  \\
			4: &h_2 &h_3 &h_4 &h_1 &&0  &  \frac{1}{4}  &  \frac{5}{12}  &  \frac{1}{3}  \\
		\end{array}
		$\\[0.5em]
		$
		\begin{array}{cccccccccc}
			&&&&&&h_1 &h_2 &h_3 &h_4\\\cmidrule{7-10}
			1: &h_2 &h_1 &h_3 &h_4 &&\frac{1}{2}  &  \frac{1}{4}  & \cellcolor{grey} 0  & \cellcolor{grey} \frac{1}{4}   \\
			2: &h_2 &h_1 &h_3 &h_4 &&\frac{1}{2}  &  \frac{1}{4}  & \cellcolor{grey} 0  & \cellcolor{grey} \frac{1}{4}  \\
			3: &h_2 &h_3 &h_4 &h_1 &&0  &  \frac{1}{4}  & \cellcolor{grey} \frac{1}{2}  & \cellcolor{grey} \frac{1}{4}  \\
			4: &h_2 &h_3 &h_4 &h_1 &&0  &  \frac{1}{4}  & \cellcolor{grey} \frac{1}{2}  & \cellcolor{grey} \frac{1}{4}  \\
		\end{array}
		$\\[0.5em]
		$
		\begin{array}{cccccccccc}
			&&&&&&h_1 & h_2 & h_3 & h_4\\\cmidrule{7-10}
			1: &h_2 &h_1 &h_3 &h_4 &&\frac{2}{3} &  \frac{1}{4}  &  0  &  \frac{1}{12}   \\
			2: &h_2 &h_3 &h_1 &h_4 &&\cellcolor{grey} \frac{1}{6}  &  \frac{1}{4}  &  \frac{1}{3}  & \cellcolor{grey} \frac{1}{4}  \\
			3: &h_2 &h_3 &h_4 &h_1 &&\cellcolor{grey}  \frac{1}{12}  &  \frac{1}{4}  &  \frac{1}{3}  & \cellcolor{grey} \frac{1}{3}  \\
			4: &h_2 &h_3 &h_4 &h_1 &&\cellcolor{grey}  \frac{1}{12}  &  \frac{1}{4}  &  \frac{1}{3}  & \cellcolor{grey} \frac{1}{3}  \\
		\end{array}
		$
		\caption{The five canonical profiles are the only canonical profiles for which the proposed rule returns a different output than \rsd. Furthermore, only entries marked in gray differ from \rsd. The rule also satisfies symmetry within domain \chch.}
		\label{fig:chch}
	\end{figure}

	We suspect that the characterization cannot hold in many subdomains of \domain and prove this exemplarily for the following subdomain proposed by \citet{ChCh16a}.
	
	Consider the subdomain \chch where all agents have the same ranking over all but one house. This domain is rich enough for the impossibility of equal treatment of equals, strategyproofness, and ordinal efficiency by \citet{BoMo01a}.
	In this domain, \rsd is \emph{not} the only rule satisfying equal treatment of equals, strategyproofness, and \emph{ex post} efficiency for $n=4$.
	We identified an alternative rule using quadratic programming that has the maximal $\ell_2$-distance to \rsd when considering the sum over all profiles of the $\ell_2$-distance of the different assignments.
	Furthermore, the rule satisfies symmetry on the subdomain, i.e., profiles that are in the same equivalence class as given profiles have the same assignment permuted accordingly.

	\begin{restatable}{theorem}{chchtheorem} \label{thm:chch}
		\rsd is not characterized by equal treatment of equals, \emph{ex post} efficiency, and strategyproofness in the domain \chch.
	\end{restatable}

\begin{proof}
The rule defined in \Cref{fig:chch} satisfies all three axioms and was found using quadratic programming. It is equal to RSD on all canonical profiles except the five shown in \Cref{fig:chch}. Profiles in the same equivalence class receive the same random assignment permuted accordingly.
\end{proof}

	\section{Conclusion}\label{sec: results}
	
	The current state of \rsd characterizations via equal treatment of equals, \emph{ex post} efficiency, and strategyproofness for small $n$ is summarized in \Cref{fig:results}.
	The first characterization for $n = 3$ was shown by \citet{BoMo01a}.
	In their proof, they use a lemma that is based on a weakening of support efficiency.
	Since \emph{ex post} efficiency and support efficiency are equivalent for $n = 3$,
	full \emph{ex post} efficiency is not required for $n = 3$.
	
	Recently, \citet{San22a} has shown via a computer-aided proof that the characterization holds for $n \le 4$ using symmetry, support efficiency, and localizedness.

It remains an open problem whether a characterization of \rsd via \emph{ex post} efficiency, strategyproofness and equal treatment of equals holds for arbitrary $n$. On the one hand, our results suggest that such a characterization might indeed hold, even when weakening efficiency and strategyproofness and without additionally demanding symmetry. In fact, the weaker axioms, in particular support efficiency instead of \emph{ex post} efficiency, seem to be a lot easier to handle for computers as well as humans.   
On the other hand, in case the characterization does not hold, our results show that another \emph{ex post} efficient and strategyproof rule that treats equals equally can only differ from \rsd when $n\geq 6$, casting doubt on the existence of a closed-form representation of any such rule.

We hope that the linear algebraic interpretation of the problem presented in this paper will prove beneficial for a complete characterization of \rsd.
Computer-generated counterexamples show that two alternative natural approaches for proving the characterization (using deterministic extreme points or restricted domains of preferences) are inadequate.

	\begin{figure}[tbp]\centering
		\begin{tabular}{rccl}
			 & Extra Condition & Strategyproofness & Source \\\midrule
			$n \le 3$ & --- & strategyproofness & \citet{BoMo01a} \\
			$n \le 4$ & symmetry & only localizedness & \citet{San22a} \\
		\end{tabular}
		\caption{Overview of characterizations of \rsd via equal treatment of equals, support efficiency, and strategyproofness for small $n$. It is open whether \emph{ex post} efficiency and nonperverseness are required for larger $n$.}
		\label{fig:results}
	\end{figure}

 \section*{Acknowledgments}

 {This work was supported by the Deutsche Forschungsgemeinschaft under grants \mbox{BR 2312/11-2} and \mbox{BR 2312/12-1}.
 We thank Itai Ashlagi, Sophie Bade, Florian Brandl, Chris Dong, Hervé Moulin, Marek Pycia, Fedor Sandomirskiy, and Omer Tamuz for helpful discussions.

\appendix

\section{Non-decomposable strategyproof and \emph{ex post} efficient rule}\label{app:sp-ef-rule}

The following assignment rule is strategyproof and \emph{ex post} efficient but cannot be represented as a convex combination of deterministic rules that satisfy these axioms.
The rule behaves similarly to \rsd in many profiles and like a serial dictatorship in others.
To highlight this, we marked all entries in which the rule disagrees with \rsd in gray.
Finally, we note that the rule only returns the probability values $0, \nicefrac{1}{3}, \nicefrac{2}{3}, 1$ in contrast to \rsd and serial dictatorships.

\vspace{1em}
\arraycolsep=2pt
\def\arraystretch{1.2}
\centering
\tiny
\noindent$
\begin{array}{ccccccc}
	&&&& h_1 & h_2 & h_3 \\\cmidrule{5-7}
	1: & h_1 & h_2 & h_3 &\nicefrac{1}{3} & \nicefrac{1}{3} & \nicefrac{1}{3} \\
	2: & h_1 & h_2 & h_3 &\nicefrac{1}{3} & \nicefrac{1}{3} & \nicefrac{1}{3} \\
	3: & h_1 & h_2 & h_3 &\nicefrac{1}{3} & \nicefrac{1}{3} & \nicefrac{1}{3} 
\end{array}
$\hspace{2em}$
\begin{array}{ccccccc}
	&&&& h_1 & h_2 & h_3 \\\cmidrule{5-7}
	1: & h_1 & h_2 & h_3 &\cellcolor{grey} \nicefrac{2}{3} & \cellcolor{grey} \nicefrac{1}{3} & \cellcolor{grey} 0 \\
	2: & h_1 & h_2 & h_3 &\cellcolor{grey} 0 & \cellcolor{grey} \nicefrac{2}{3} & \cellcolor{grey} \nicefrac{1}{3} \\
	3: & h_1 & h_3 & h_2 &\nicefrac{1}{3} & 0 & \nicefrac{2}{3} 
\end{array}
$\hspace{2em}$
\begin{array}{ccccccc}
	&&&& h_1 & h_2 & h_3 \\\cmidrule{5-7}
	1: & h_1 & h_2 & h_3 &\cellcolor{grey} \nicefrac{1}{3} & \cellcolor{grey} \nicefrac{1}{3} & \nicefrac{1}{3} \\
	2: & h_1 & h_2 & h_3 &\cellcolor{grey} \nicefrac{2}{3} & \cellcolor{grey} 0 & \nicefrac{1}{3} \\
	3: & h_2 & h_1 & h_3 &0 & \nicefrac{2}{3} & \nicefrac{1}{3} 
\end{array}
$\hspace{2em}$
\begin{array}{ccccccc}
	&&&& h_1 & h_2 & h_3 \\\cmidrule{5-7}
	1: & h_1 & h_2 & h_3 &\cellcolor{grey} \nicefrac{1}{3} & \cellcolor{grey} 0 & \cellcolor{grey} \nicefrac{2}{3} \\
	2: & h_1 & h_2 & h_3 &\cellcolor{grey} \nicefrac{2}{3} & \cellcolor{grey} \nicefrac{1}{3} & \cellcolor{grey} 0 \\
	3: & h_2 & h_3 & h_1 &0 & \nicefrac{2}{3} & \nicefrac{1}{3} 
\end{array}
$\\\vspace{1em}$
\begin{array}{ccccccc}
	&&&& h_1 & h_2 & h_3 \\\cmidrule{5-7}
	1: & h_1 & h_2 & h_3 &\cellcolor{grey} 0 & \cellcolor{grey} 1 & 0 \\
	2: & h_1 & h_2 & h_3 &\cellcolor{grey} 1 & \cellcolor{grey} 0 & 0 \\
	3: & h_3 & h_1 & h_2 &0 & 0 & 1 
\end{array}
$\hspace{2em}$
\begin{array}{ccccccc}
	&&&& h_1 & h_2 & h_3 \\\cmidrule{5-7}
	1: & h_1 & h_2 & h_3 &\cellcolor{grey} \nicefrac{2}{3} & \cellcolor{grey} \nicefrac{1}{3} & 0 \\
	2: & h_1 & h_2 & h_3 &\cellcolor{grey} \nicefrac{1}{3} & \cellcolor{grey} \nicefrac{2}{3} & 0 \\
	3: & h_3 & h_2 & h_1 &0 & 0 & 1 
\end{array}
$\hspace{2em}$
\begin{array}{ccccccc}
	&&&& h_1 & h_2 & h_3 \\\cmidrule{5-7}
	1: & h_1 & h_2 & h_3 &\cellcolor{grey} 0 & \cellcolor{grey} 1 & \cellcolor{grey} 0 \\
	2: & h_1 & h_3 & h_2 &\nicefrac{1}{3} & 0 & \nicefrac{2}{3} \\
	3: & h_1 & h_2 & h_3 &\cellcolor{grey} \nicefrac{2}{3} & \cellcolor{grey} 0 & \cellcolor{grey} \nicefrac{1}{3} 
\end{array}
$\hspace{2em}$
\begin{array}{ccccccc}
	&&&& h_1 & h_2 & h_3 \\\cmidrule{5-7}
	1: & h_1 & h_2 & h_3 &\nicefrac{1}{3} & \nicefrac{2}{3} & 0 \\
	2: & h_1 & h_3 & h_2 &\cellcolor{grey} 0 & \cellcolor{grey} \nicefrac{1}{3} & \cellcolor{grey} \nicefrac{2}{3} \\
	3: & h_1 & h_3 & h_2 &\cellcolor{grey} \nicefrac{2}{3} & \cellcolor{grey} 0 & \cellcolor{grey} \nicefrac{1}{3} 
\end{array}
$\\\vspace{1em}$
\begin{array}{ccccccc}
	&&&& h_1 & h_2 & h_3 \\\cmidrule{5-7}
	1: & h_1 & h_2 & h_3 &\cellcolor{grey} \nicefrac{1}{3} & \cellcolor{grey} \nicefrac{1}{3} & \nicefrac{1}{3} \\
	2: & h_1 & h_3 & h_2 &\cellcolor{grey} \nicefrac{2}{3} & 0 & \cellcolor{grey} \nicefrac{1}{3} \\
	3: & h_2 & h_1 & h_3 &0 & \cellcolor{grey} \nicefrac{2}{3} & \cellcolor{grey} \nicefrac{1}{3} 
\end{array}
$\hspace{2em}$
\begin{array}{ccccccc}
	&&&& h_1 & h_2 & h_3 \\\cmidrule{5-7}
	1: & h_1 & h_2 & h_3 &\cellcolor{grey} \nicefrac{1}{3} & \cellcolor{grey} \nicefrac{1}{3} & \nicefrac{1}{3} \\
	2: & h_1 & h_3 & h_2 &\cellcolor{grey} \nicefrac{2}{3} & 0 & \cellcolor{grey} \nicefrac{1}{3} \\
	3: & h_2 & h_3 & h_1 &0 & \cellcolor{grey} \nicefrac{2}{3} & \cellcolor{grey} \nicefrac{1}{3} 
\end{array}
$\hspace{2em}$
\begin{array}{ccccccc}
	&&&& h_1 & h_2 & h_3 \\\cmidrule{5-7}
	1: & h_1 & h_2 & h_3 &\cellcolor{grey} 0 & \cellcolor{grey} 1 & 0 \\
	2: & h_1 & h_3 & h_2 &\cellcolor{grey} 1 & \cellcolor{grey} 0 & \cellcolor{grey} 0 \\
	3: & h_3 & h_1 & h_2 &0 & \cellcolor{grey} 0 & \cellcolor{grey} 1 
\end{array}
$\hspace{2em}$
\begin{array}{ccccccc}
	&&&& h_1 & h_2 & h_3 \\\cmidrule{5-7}
	1: & h_1 & h_2 & h_3 &\cellcolor{grey} \nicefrac{2}{3} & \cellcolor{grey} \nicefrac{1}{3} & 0 \\
	2: & h_1 & h_3 & h_2 &\cellcolor{grey} \nicefrac{1}{3} & \cellcolor{grey} \nicefrac{2}{3} & \cellcolor{grey} 0 \\
	3: & h_3 & h_2 & h_1 &0 & \cellcolor{grey} 0 & \cellcolor{grey} 1 
\end{array}
$\\\vspace{1em}$
\begin{array}{ccccccc}
	&&&& h_1 & h_2 & h_3 \\\cmidrule{5-7}
	1: & h_1 & h_2 & h_3 &\cellcolor{grey} \nicefrac{2}{3} & \cellcolor{grey} 0 & \nicefrac{1}{3} \\
	2: & h_2 & h_1 & h_3 &0 & \nicefrac{2}{3} & \nicefrac{1}{3} \\
	3: & h_1 & h_2 & h_3 &\cellcolor{grey} \nicefrac{1}{3} & \cellcolor{grey} \nicefrac{1}{3} & \nicefrac{1}{3} 
\end{array}
$\hspace{2em}$
\begin{array}{ccccccc}
	&&&& h_1 & h_2 & h_3 \\\cmidrule{5-7}
	1: & h_1 & h_2 & h_3 &\cellcolor{grey} \nicefrac{2}{3} & \cellcolor{grey} \nicefrac{1}{3} & \cellcolor{grey} 0 \\
	2: & h_2 & h_1 & h_3 &0 & \cellcolor{grey} \nicefrac{2}{3} & \cellcolor{grey} \nicefrac{1}{3} \\
	3: & h_1 & h_3 & h_2 &\cellcolor{grey} \nicefrac{1}{3} & 0 & \cellcolor{grey} \nicefrac{2}{3} 
\end{array}
$\hspace{2em}$
\begin{array}{ccccccc}
	&&&& h_1 & h_2 & h_3 \\\cmidrule{5-7}
	1: & h_1 & h_2 & h_3 &\nicefrac{2}{3} & 0 & \nicefrac{1}{3} \\
	2: & h_2 & h_1 & h_3 &\cellcolor{grey} 0 & \cellcolor{grey} \nicefrac{2}{3} & \nicefrac{1}{3} \\
	3: & h_2 & h_1 & h_3 &\cellcolor{grey} \nicefrac{1}{3} & \cellcolor{grey} \nicefrac{1}{3} & \nicefrac{1}{3} 
\end{array}
$\hspace{2em}$
\begin{array}{ccccccc}
	&&&& h_1 & h_2 & h_3 \\\cmidrule{5-7}
	1: & h_1 & h_2 & h_3 &\cellcolor{grey} \nicefrac{2}{3} & 0 & \cellcolor{grey} \nicefrac{1}{3} \\
	2: & h_2 & h_1 & h_3 &\cellcolor{grey} \nicefrac{1}{3} & \cellcolor{grey} \nicefrac{2}{3} & \cellcolor{grey} 0 \\
	3: & h_2 & h_3 & h_1 &0 & \cellcolor{grey} \nicefrac{1}{3} & \cellcolor{grey} \nicefrac{2}{3} 
\end{array}
$\\\vspace{1em}$
\begin{array}{ccccccc}
	&&&& h_1 & h_2 & h_3 \\\cmidrule{5-7}
	1: & h_1 & h_2 & h_3 &1 & 0 & 0 \\
	2: & h_2 & h_1 & h_3 &0 & 1 & 0 \\
	3: & h_3 & h_1 & h_2 &0 & 0 & 1 
\end{array}
$\hspace{2em}$
\begin{array}{ccccccc}
	&&&& h_1 & h_2 & h_3 \\\cmidrule{5-7}
	1: & h_1 & h_2 & h_3 &1 & 0 & 0 \\
	2: & h_2 & h_1 & h_3 &0 & 1 & 0 \\
	3: & h_3 & h_2 & h_1 &0 & 0 & 1 
\end{array}
$\hspace{2em}$
\begin{array}{ccccccc}
	&&&& h_1 & h_2 & h_3 \\\cmidrule{5-7}
	1: & h_1 & h_2 & h_3 &\cellcolor{grey} \nicefrac{1}{3} & \cellcolor{grey} \nicefrac{1}{3} & \nicefrac{1}{3} \\
	2: & h_2 & h_3 & h_1 &0 & \nicefrac{2}{3} & \nicefrac{1}{3} \\
	3: & h_1 & h_2 & h_3 &\cellcolor{grey} \nicefrac{2}{3} & \cellcolor{grey} 0 & \nicefrac{1}{3} 
\end{array}
$\hspace{2em}$
\begin{array}{ccccccc}
	&&&& h_1 & h_2 & h_3 \\\cmidrule{5-7}
	1: & h_1 & h_2 & h_3 &\cellcolor{grey} \nicefrac{1}{3} & \cellcolor{grey} \nicefrac{1}{3} & \nicefrac{1}{3} \\
	2: & h_2 & h_3 & h_1 &0 & \cellcolor{grey} \nicefrac{2}{3} & \cellcolor{grey} \nicefrac{1}{3} \\
	3: & h_1 & h_3 & h_2 &\cellcolor{grey} \nicefrac{2}{3} & 0 & \cellcolor{grey} \nicefrac{1}{3} 
\end{array}
$\\\vspace{1em}$
\begin{array}{ccccccc}
	&&&& h_1 & h_2 & h_3 \\\cmidrule{5-7}
	1: & h_1 & h_2 & h_3 &\cellcolor{grey} \nicefrac{2}{3} & 0 & \cellcolor{grey} \nicefrac{1}{3} \\
	2: & h_2 & h_3 & h_1 &0 & \cellcolor{grey} \nicefrac{2}{3} & \cellcolor{grey} \nicefrac{1}{3} \\
	3: & h_2 & h_1 & h_3 &\cellcolor{grey} \nicefrac{1}{3} & \cellcolor{grey} \nicefrac{1}{3} & \nicefrac{1}{3} 
\end{array}
$\hspace{2em}$
\begin{array}{ccccccc}
	&&&& h_1 & h_2 & h_3 \\\cmidrule{5-7}
	1: & h_1 & h_2 & h_3 &1 & 0 & 0 \\
	2: & h_2 & h_3 & h_1 &0 & \cellcolor{grey} \nicefrac{2}{3} & \cellcolor{grey} \nicefrac{1}{3} \\
	3: & h_2 & h_3 & h_1 &0 & \cellcolor{grey} \nicefrac{1}{3} & \cellcolor{grey} \nicefrac{2}{3} 
\end{array}
$\hspace{2em}$
\begin{array}{ccccccc}
	&&&& h_1 & h_2 & h_3 \\\cmidrule{5-7}
	1: & h_1 & h_2 & h_3 &1 & 0 & 0 \\
	2: & h_2 & h_3 & h_1 &0 & 1 & 0 \\
	3: & h_3 & h_1 & h_2 &0 & 0 & 1 
\end{array}
$\hspace{2em}$
\begin{array}{ccccccc}
	&&&& h_1 & h_2 & h_3 \\\cmidrule{5-7}
	1: & h_1 & h_2 & h_3 &1 & 0 & 0 \\
	2: & h_2 & h_3 & h_1 &0 & 1 & 0 \\
	3: & h_3 & h_2 & h_1 &0 & 0 & 1 
\end{array}
$\\\vspace{1em}$
\begin{array}{ccccccc}
	&&&& h_1 & h_2 & h_3 \\\cmidrule{5-7}
	1: & h_1 & h_2 & h_3 &\cellcolor{grey} \nicefrac{2}{3} & \cellcolor{grey} \nicefrac{1}{3} & 0 \\
	2: & h_3 & h_1 & h_2 &0 & 0 & 1 \\
	3: & h_1 & h_2 & h_3 &\cellcolor{grey} \nicefrac{1}{3} & \cellcolor{grey} \nicefrac{2}{3} & 0 
\end{array}
$\hspace{2em}$
\begin{array}{ccccccc}
	&&&& h_1 & h_2 & h_3 \\\cmidrule{5-7}
	1: & h_1 & h_2 & h_3 &\cellcolor{grey} \nicefrac{2}{3} & \cellcolor{grey} \nicefrac{1}{3} & 0 \\
	2: & h_3 & h_1 & h_2 &0 & \cellcolor{grey} \nicefrac{1}{3} & \cellcolor{grey} \nicefrac{2}{3} \\
	3: & h_1 & h_3 & h_2 &\cellcolor{grey} \nicefrac{1}{3} & \nicefrac{1}{3} & \cellcolor{grey} \nicefrac{1}{3} 
\end{array}
$\hspace{2em}$
\begin{array}{ccccccc}
	&&&& h_1 & h_2 & h_3 \\\cmidrule{5-7}
	1: & h_1 & h_2 & h_3 &1 & 0 & 0 \\
	2: & h_3 & h_1 & h_2 &0 & 0 & 1 \\
	3: & h_2 & h_1 & h_3 &0 & 1 & 0 
\end{array}
$\hspace{2em}$
\begin{array}{ccccccc}
	&&&& h_1 & h_2 & h_3 \\\cmidrule{5-7}
	1: & h_1 & h_2 & h_3 &1 & 0 & 0 \\
	2: & h_3 & h_1 & h_2 &0 & 0 & 1 \\
	3: & h_2 & h_3 & h_1 &0 & 1 & 0 
\end{array}
$\\\vspace{1em}$
\begin{array}{ccccccc}
	&&&& h_1 & h_2 & h_3 \\\cmidrule{5-7}
	1: & h_1 & h_2 & h_3 &\cellcolor{grey} \nicefrac{1}{3} & \cellcolor{grey} \nicefrac{2}{3} & 0 \\
	2: & h_3 & h_1 & h_2 &\cellcolor{grey} \nicefrac{2}{3} & \cellcolor{grey} 0 & \cellcolor{grey} \nicefrac{1}{3} \\
	3: & h_3 & h_1 & h_2 &\cellcolor{grey} 0 & \nicefrac{1}{3} & \cellcolor{grey} \nicefrac{2}{3} 
\end{array}
$\hspace{2em}$
\begin{array}{ccccccc}
	&&&& h_1 & h_2 & h_3 \\\cmidrule{5-7}
	1: & h_1 & h_2 & h_3 &\cellcolor{grey} 1 & \cellcolor{grey} 0 & 0 \\
	2: & h_3 & h_1 & h_2 &\cellcolor{grey} 0 & \cellcolor{grey} \nicefrac{2}{3} & \cellcolor{grey} \nicefrac{1}{3} \\
	3: & h_3 & h_2 & h_1 &0 & \cellcolor{grey} \nicefrac{1}{3} & \cellcolor{grey} \nicefrac{2}{3} 
\end{array}
$\hspace{2em}$
\begin{array}{ccccccc}
	&&&& h_1 & h_2 & h_3 \\\cmidrule{5-7}
	1: & h_1 & h_2 & h_3 &\cellcolor{grey} \nicefrac{2}{3} & \cellcolor{grey} \nicefrac{1}{3} & 0 \\
	2: & h_3 & h_2 & h_1 &0 & 0 & 1 \\
	3: & h_1 & h_2 & h_3 &\cellcolor{grey} \nicefrac{1}{3} & \cellcolor{grey} \nicefrac{2}{3} & 0 
\end{array}
$\hspace{2em}$
\begin{array}{ccccccc}
	&&&& h_1 & h_2 & h_3 \\\cmidrule{5-7}
	1: & h_1 & h_2 & h_3 &\cellcolor{grey} \nicefrac{2}{3} & \cellcolor{grey} \nicefrac{1}{3} & 0 \\
	2: & h_3 & h_2 & h_1 &0 & \cellcolor{grey} \nicefrac{1}{3} & \cellcolor{grey} \nicefrac{2}{3} \\
	3: & h_1 & h_3 & h_2 &\cellcolor{grey} \nicefrac{1}{3} & \nicefrac{1}{3} & \cellcolor{grey} \nicefrac{1}{3} 
\end{array}
$\\\vspace{1em}$
\begin{array}{ccccccc}
	&&&& h_1 & h_2 & h_3 \\\cmidrule{5-7}
	1: & h_1 & h_2 & h_3 &1 & 0 & 0 \\
	2: & h_3 & h_2 & h_1 &0 & 0 & 1 \\
	3: & h_2 & h_1 & h_3 &0 & 1 & 0 
\end{array}
$\hspace{2em}$
\begin{array}{ccccccc}
	&&&& h_1 & h_2 & h_3 \\\cmidrule{5-7}
	1: & h_1 & h_2 & h_3 &1 & 0 & 0 \\
	2: & h_3 & h_2 & h_1 &0 & 0 & 1 \\
	3: & h_2 & h_3 & h_1 &0 & 1 & 0 
\end{array}
$\hspace{2em}$
\begin{array}{ccccccc}
	&&&& h_1 & h_2 & h_3 \\\cmidrule{5-7}
	1: & h_1 & h_2 & h_3 &\cellcolor{grey} 1 & \cellcolor{grey} 0 & 0 \\
	2: & h_3 & h_2 & h_1 &0 & \cellcolor{grey} \nicefrac{2}{3} & \cellcolor{grey} \nicefrac{1}{3} \\
	3: & h_3 & h_1 & h_2 &\cellcolor{grey} 0 & \nicefrac{1}{3} & \cellcolor{grey} \nicefrac{2}{3} 
\end{array}
$\hspace{2em}$
\begin{array}{ccccccc}
	&&&& h_1 & h_2 & h_3 \\\cmidrule{5-7}
	1: & h_1 & h_2 & h_3 &1 & 0 & 0 \\
	2: & h_3 & h_2 & h_1 &0 & \cellcolor{grey} \nicefrac{2}{3} & \cellcolor{grey} \nicefrac{1}{3} \\
	3: & h_3 & h_2 & h_1 &0 & \cellcolor{grey} \nicefrac{1}{3} & \cellcolor{grey} \nicefrac{2}{3} 
\end{array}
$\\\vspace{1em}$
\begin{array}{ccccccc}
	&&&& h_1 & h_2 & h_3 \\\cmidrule{5-7}
	1: & h_1 & h_3 & h_2 &\nicefrac{1}{3} & 0 & \nicefrac{2}{3} \\
	2: & h_1 & h_2 & h_3 &\cellcolor{grey} \nicefrac{2}{3} & \cellcolor{grey} \nicefrac{1}{3} & \cellcolor{grey} 0 \\
	3: & h_1 & h_2 & h_3 &\cellcolor{grey} 0 & \cellcolor{grey} \nicefrac{2}{3} & \cellcolor{grey} \nicefrac{1}{3} 
\end{array}
$\hspace{2em}$
\begin{array}{ccccccc}
	&&&& h_1 & h_2 & h_3 \\\cmidrule{5-7}
	1: & h_1 & h_3 & h_2 &\cellcolor{grey} \nicefrac{2}{3} & \cellcolor{grey} 0 & \cellcolor{grey} \nicefrac{1}{3} \\
	2: & h_1 & h_2 & h_3 &\nicefrac{1}{3} & \nicefrac{2}{3} & 0 \\
	3: & h_1 & h_3 & h_2 &\cellcolor{grey} 0 & \cellcolor{grey} \nicefrac{1}{3} & \cellcolor{grey} \nicefrac{2}{3} 
\end{array}
$\hspace{2em}$
\begin{array}{ccccccc}
	&&&& h_1 & h_2 & h_3 \\\cmidrule{5-7}
	1: & h_1 & h_3 & h_2 &\cellcolor{grey} \nicefrac{1}{3} & 0 & \cellcolor{grey} \nicefrac{2}{3} \\
	2: & h_1 & h_2 & h_3 &\cellcolor{grey} \nicefrac{2}{3} & \cellcolor{grey} \nicefrac{1}{3} & \cellcolor{grey} 0 \\
	3: & h_2 & h_1 & h_3 &0 & \cellcolor{grey} \nicefrac{2}{3} & \cellcolor{grey} \nicefrac{1}{3} 
\end{array}
$\hspace{2em}$
\begin{array}{ccccccc}
	&&&& h_1 & h_2 & h_3 \\\cmidrule{5-7}
	1: & h_1 & h_3 & h_2 &\cellcolor{grey} \nicefrac{1}{3} & 0 & \cellcolor{grey} \nicefrac{2}{3} \\
	2: & h_1 & h_2 & h_3 &\cellcolor{grey} \nicefrac{2}{3} & \cellcolor{grey} \nicefrac{1}{3} & \cellcolor{grey} 0 \\
	3: & h_2 & h_3 & h_1 &0 & \cellcolor{grey} \nicefrac{2}{3} & \cellcolor{grey} \nicefrac{1}{3} 
\end{array}
$\\\vspace{1em}$
\begin{array}{ccccccc}
	&&&& h_1 & h_2 & h_3 \\\cmidrule{5-7}
	1: & h_1 & h_3 & h_2 &\cellcolor{grey} 0 & \cellcolor{grey} \nicefrac{2}{3} & \cellcolor{grey} \nicefrac{1}{3} \\
	2: & h_1 & h_2 & h_3 &\cellcolor{grey} 1 & \cellcolor{grey} 0 & 0 \\
	3: & h_3 & h_1 & h_2 &0 & \cellcolor{grey} \nicefrac{1}{3} & \cellcolor{grey} \nicefrac{2}{3} 
\end{array}
$\hspace{2em}$
\begin{array}{ccccccc}
	&&&& h_1 & h_2 & h_3 \\\cmidrule{5-7}
	1: & h_1 & h_3 & h_2 &\cellcolor{grey} \nicefrac{2}{3} & \cellcolor{grey} 0 & \cellcolor{grey} \nicefrac{1}{3} \\
	2: & h_1 & h_2 & h_3 &\cellcolor{grey} \nicefrac{1}{3} & \cellcolor{grey} \nicefrac{2}{3} & 0 \\
	3: & h_3 & h_2 & h_1 &0 & \cellcolor{grey} \nicefrac{1}{3} & \cellcolor{grey} \nicefrac{2}{3} 
\end{array}
$\hspace{2em}$
\begin{array}{ccccccc}
	&&&& h_1 & h_2 & h_3 \\\cmidrule{5-7}
	1: & h_1 & h_3 & h_2 &\cellcolor{grey} 0 & \cellcolor{grey} 0 & \cellcolor{grey} 1 \\
	2: & h_1 & h_3 & h_2 &\cellcolor{grey} \nicefrac{2}{3} & \cellcolor{grey} \nicefrac{1}{3} & \cellcolor{grey} 0 \\
	3: & h_1 & h_2 & h_3 &\nicefrac{1}{3} & \nicefrac{2}{3} & 0 
\end{array}
$\hspace{2em}$
\begin{array}{ccccccc}
	&&&& h_1 & h_2 & h_3 \\\cmidrule{5-7}
	1: & h_1 & h_3 & h_2 &\nicefrac{1}{3} & \nicefrac{1}{3} & \nicefrac{1}{3} \\
	2: & h_1 & h_3 & h_2 &\nicefrac{1}{3} & \cellcolor{grey} \nicefrac{2}{3} & \cellcolor{grey} 0 \\
	3: & h_1 & h_3 & h_2 &\nicefrac{1}{3} & \cellcolor{grey} 0 & \cellcolor{grey} \nicefrac{2}{3} 
\end{array}
$\\\vspace{1em}$
\begin{array}{ccccccc}
	&&&& h_1 & h_2 & h_3 \\\cmidrule{5-7}
	1: & h_1 & h_3 & h_2 &\cellcolor{grey} \nicefrac{1}{3} & 0 & \cellcolor{grey} \nicefrac{2}{3} \\
	2: & h_1 & h_3 & h_2 &\cellcolor{grey} \nicefrac{2}{3} & 0 & \cellcolor{grey} \nicefrac{1}{3} \\
	3: & h_2 & h_1 & h_3 &0 & 1 & 0 
\end{array}
$\hspace{2em}$
\begin{array}{ccccccc}
	&&&& h_1 & h_2 & h_3 \\\cmidrule{5-7}
	1: & h_1 & h_3 & h_2 &\cellcolor{grey} \nicefrac{1}{3} & 0 & \cellcolor{grey} \nicefrac{2}{3} \\
	2: & h_1 & h_3 & h_2 &\cellcolor{grey} \nicefrac{2}{3} & 0 & \cellcolor{grey} \nicefrac{1}{3} \\
	3: & h_2 & h_3 & h_1 &0 & 1 & 0 
\end{array}
$\hspace{2em}$
\begin{array}{ccccccc}
	&&&& h_1 & h_2 & h_3 \\\cmidrule{5-7}
	1: & h_1 & h_3 & h_2 &\cellcolor{grey} 0 & \cellcolor{grey} 1 & \cellcolor{grey} 0 \\
	2: & h_1 & h_3 & h_2 &\cellcolor{grey} 1 & \cellcolor{grey} 0 & \cellcolor{grey} 0 \\
	3: & h_3 & h_1 & h_2 &0 & \cellcolor{grey} 0 & \cellcolor{grey} 1 
\end{array}
$\hspace{2em}$
\begin{array}{ccccccc}
	&&&& h_1 & h_2 & h_3 \\\cmidrule{5-7}
	1: & h_1 & h_3 & h_2 &\cellcolor{grey} \nicefrac{2}{3} & \nicefrac{1}{3} & \cellcolor{grey} 0 \\
	2: & h_1 & h_3 & h_2 &\cellcolor{grey} \nicefrac{1}{3} & \cellcolor{grey} \nicefrac{2}{3} & \cellcolor{grey} 0 \\
	3: & h_3 & h_2 & h_1 &0 & \cellcolor{grey} 0 & \cellcolor{grey} 1 
\end{array}
$\\\vspace{1em}$
\begin{array}{ccccccc}
	&&&& h_1 & h_2 & h_3 \\\cmidrule{5-7}
	1: & h_1 & h_3 & h_2 &\cellcolor{grey} \nicefrac{2}{3} & 0 & \cellcolor{grey} \nicefrac{1}{3} \\
	2: & h_2 & h_1 & h_3 &0 & \cellcolor{grey} 1 & \cellcolor{grey} 0 \\
	3: & h_1 & h_2 & h_3 &\cellcolor{grey} \nicefrac{1}{3} & \cellcolor{grey} 0 & \cellcolor{grey} \nicefrac{2}{3} 
\end{array}
$\hspace{2em}$
\begin{array}{ccccccc}
	&&&& h_1 & h_2 & h_3 \\\cmidrule{5-7}
	1: & h_1 & h_3 & h_2 &\cellcolor{grey} \nicefrac{2}{3} & 0 & \cellcolor{grey} \nicefrac{1}{3} \\
	2: & h_2 & h_1 & h_3 &0 & 1 & 0 \\
	3: & h_1 & h_3 & h_2 &\cellcolor{grey} \nicefrac{1}{3} & 0 & \cellcolor{grey} \nicefrac{2}{3} 
\end{array}
$\hspace{2em}$
\begin{array}{ccccccc}
	&&&& h_1 & h_2 & h_3 \\\cmidrule{5-7}
	1: & h_1 & h_3 & h_2 &\nicefrac{2}{3} & 0 & \nicefrac{1}{3} \\
	2: & h_2 & h_1 & h_3 &\cellcolor{grey} \nicefrac{1}{3} & \cellcolor{grey} \nicefrac{2}{3} & \cellcolor{grey} 0 \\
	3: & h_2 & h_1 & h_3 &\cellcolor{grey} 0 & \cellcolor{grey} \nicefrac{1}{3} & \cellcolor{grey} \nicefrac{2}{3} 
\end{array}
$\hspace{2em}$
\begin{array}{ccccccc}
	&&&& h_1 & h_2 & h_3 \\\cmidrule{5-7}
	1: & h_1 & h_3 & h_2 &\cellcolor{grey} \nicefrac{2}{3} & 0 & \cellcolor{grey} \nicefrac{1}{3} \\
	2: & h_2 & h_1 & h_3 &\cellcolor{grey} \nicefrac{1}{3} & \cellcolor{grey} \nicefrac{2}{3} & \cellcolor{grey} 0 \\
	3: & h_2 & h_3 & h_1 &0 & \cellcolor{grey} \nicefrac{1}{3} & \cellcolor{grey} \nicefrac{2}{3} 
\end{array}
$\\\vspace{1em}$
\begin{array}{ccccccc}
	&&&& h_1 & h_2 & h_3 \\\cmidrule{5-7}
	1: & h_1 & h_3 & h_2 &1 & 0 & 0 \\
	2: & h_2 & h_1 & h_3 &0 & 1 & 0 \\
	3: & h_3 & h_1 & h_2 &0 & 0 & 1 
\end{array}
$\hspace{2em}$
\begin{array}{ccccccc}
	&&&& h_1 & h_2 & h_3 \\\cmidrule{5-7}
	1: & h_1 & h_3 & h_2 &1 & 0 & 0 \\
	2: & h_2 & h_1 & h_3 &0 & 1 & 0 \\
	3: & h_3 & h_2 & h_1 &0 & 0 & 1 
\end{array}
$\hspace{2em}$
\begin{array}{ccccccc}
	&&&& h_1 & h_2 & h_3 \\\cmidrule{5-7}
	1: & h_1 & h_3 & h_2 &\cellcolor{grey} \nicefrac{1}{3} & 0 & \cellcolor{grey} \nicefrac{2}{3} \\
	2: & h_2 & h_3 & h_1 &0 & \cellcolor{grey} 1 & \cellcolor{grey} 0 \\
	3: & h_1 & h_2 & h_3 &\cellcolor{grey} \nicefrac{2}{3} & \cellcolor{grey} 0 & \nicefrac{1}{3} 
\end{array}
$\hspace{2em}$
\begin{array}{ccccccc}
	&&&& h_1 & h_2 & h_3 \\\cmidrule{5-7}
	1: & h_1 & h_3 & h_2 &\cellcolor{grey} \nicefrac{1}{3} & 0 & \cellcolor{grey} \nicefrac{2}{3} \\
	2: & h_2 & h_3 & h_1 &0 & 1 & 0 \\
	3: & h_1 & h_3 & h_2 &\cellcolor{grey} \nicefrac{2}{3} & 0 & \cellcolor{grey} \nicefrac{1}{3} 
\end{array}
$\\\vspace{1em}$
\begin{array}{ccccccc}
	&&&& h_1 & h_2 & h_3 \\\cmidrule{5-7}
	1: & h_1 & h_3 & h_2 &\cellcolor{grey} \nicefrac{2}{3} & 0 & \cellcolor{grey} \nicefrac{1}{3} \\
	2: & h_2 & h_3 & h_1 &0 & \cellcolor{grey} \nicefrac{2}{3} & \cellcolor{grey} \nicefrac{1}{3} \\
	3: & h_2 & h_1 & h_3 &\cellcolor{grey} \nicefrac{1}{3} & \cellcolor{grey} \nicefrac{1}{3} & \nicefrac{1}{3} 
\end{array}
$\hspace{2em}$
\begin{array}{ccccccc}
	&&&& h_1 & h_2 & h_3 \\\cmidrule{5-7}
	1: & h_1 & h_3 & h_2 &1 & 0 & 0 \\
	2: & h_2 & h_3 & h_1 &0 & \cellcolor{grey} \nicefrac{2}{3} & \cellcolor{grey} \nicefrac{1}{3} \\
	3: & h_2 & h_3 & h_1 &0 & \cellcolor{grey} \nicefrac{1}{3} & \cellcolor{grey} \nicefrac{2}{3} 
\end{array}
$\hspace{2em}$
\begin{array}{ccccccc}
	&&&& h_1 & h_2 & h_3 \\\cmidrule{5-7}
	1: & h_1 & h_3 & h_2 &1 & 0 & 0 \\
	2: & h_2 & h_3 & h_1 &0 & 1 & 0 \\
	3: & h_3 & h_1 & h_2 &0 & 0 & 1 
\end{array}
$\hspace{2em}$
\begin{array}{ccccccc}
	&&&& h_1 & h_2 & h_3 \\\cmidrule{5-7}
	1: & h_1 & h_3 & h_2 &1 & 0 & 0 \\
	2: & h_2 & h_3 & h_1 &0 & 1 & 0 \\
	3: & h_3 & h_2 & h_1 &0 & 0 & 1 
\end{array}
$\\\vspace{1em}$
\begin{array}{ccccccc}
	&&&& h_1 & h_2 & h_3 \\\cmidrule{5-7}
	1: & h_1 & h_3 & h_2 &\cellcolor{grey} \nicefrac{2}{3} & \cellcolor{grey} 0 & \cellcolor{grey} \nicefrac{1}{3} \\
	2: & h_3 & h_1 & h_2 &0 & \cellcolor{grey} \nicefrac{1}{3} & \cellcolor{grey} \nicefrac{2}{3} \\
	3: & h_1 & h_2 & h_3 &\cellcolor{grey} \nicefrac{1}{3} & \cellcolor{grey} \nicefrac{2}{3} & 0 
\end{array}
$\hspace{2em}$
\begin{array}{ccccccc}
	&&&& h_1 & h_2 & h_3 \\\cmidrule{5-7}
	1: & h_1 & h_3 & h_2 &\cellcolor{grey} \nicefrac{2}{3} & \cellcolor{grey} 0 & \cellcolor{grey} \nicefrac{1}{3} \\
	2: & h_3 & h_1 & h_2 &0 & \cellcolor{grey} \nicefrac{2}{3} & \cellcolor{grey} \nicefrac{1}{3} \\
	3: & h_1 & h_3 & h_2 &\cellcolor{grey} \nicefrac{1}{3} & \nicefrac{1}{3} & \cellcolor{grey} \nicefrac{1}{3} 
\end{array}
$\hspace{2em}$
\begin{array}{ccccccc}
	&&&& h_1 & h_2 & h_3 \\\cmidrule{5-7}
	1: & h_1 & h_3 & h_2 &1 & 0 & 0 \\
	2: & h_3 & h_1 & h_2 &0 & 0 & 1 \\
	3: & h_2 & h_1 & h_3 &0 & 1 & 0 
\end{array}
$\hspace{2em}$
\begin{array}{ccccccc}
	&&&& h_1 & h_2 & h_3 \\\cmidrule{5-7}
	1: & h_1 & h_3 & h_2 &1 & 0 & 0 \\
	2: & h_3 & h_1 & h_2 &0 & 0 & 1 \\
	3: & h_2 & h_3 & h_1 &0 & 1 & 0 
\end{array}
$\\\vspace{1em}$
\begin{array}{ccccccc}
	&&&& h_1 & h_2 & h_3 \\\cmidrule{5-7}
	1: & h_1 & h_3 & h_2 &\cellcolor{grey} \nicefrac{1}{3} & \cellcolor{grey} \nicefrac{2}{3} & 0 \\
	2: & h_3 & h_1 & h_2 &\cellcolor{grey} \nicefrac{2}{3} & \cellcolor{grey} 0 & \cellcolor{grey} \nicefrac{1}{3} \\
	3: & h_3 & h_1 & h_2 &\cellcolor{grey} 0 & \nicefrac{1}{3} & \cellcolor{grey} \nicefrac{2}{3} 
\end{array}
$\hspace{2em}$
\begin{array}{ccccccc}
	&&&& h_1 & h_2 & h_3 \\\cmidrule{5-7}
	1: & h_1 & h_3 & h_2 &\cellcolor{grey} 1 & \cellcolor{grey} 0 & 0 \\
	2: & h_3 & h_1 & h_2 &\cellcolor{grey} 0 & \cellcolor{grey} \nicefrac{2}{3} & \cellcolor{grey} \nicefrac{1}{3} \\
	3: & h_3 & h_2 & h_1 &0 & \cellcolor{grey} \nicefrac{1}{3} & \cellcolor{grey} \nicefrac{2}{3} 
\end{array}
$\hspace{2em}$
\begin{array}{ccccccc}
	&&&& h_1 & h_2 & h_3 \\\cmidrule{5-7}
	1: & h_1 & h_3 & h_2 &\cellcolor{grey} \nicefrac{2}{3} & \cellcolor{grey} 0 & \cellcolor{grey} \nicefrac{1}{3} \\
	2: & h_3 & h_2 & h_1 &0 & \cellcolor{grey} \nicefrac{1}{3} & \cellcolor{grey} \nicefrac{2}{3} \\
	3: & h_1 & h_2 & h_3 &\cellcolor{grey} \nicefrac{1}{3} & \cellcolor{grey} \nicefrac{2}{3} & 0 
\end{array}
$\hspace{2em}$
\begin{array}{ccccccc}
	&&&& h_1 & h_2 & h_3 \\\cmidrule{5-7}
	1: & h_1 & h_3 & h_2 &\cellcolor{grey} \nicefrac{2}{3} & \cellcolor{grey} 0 & \cellcolor{grey} \nicefrac{1}{3} \\
	2: & h_3 & h_2 & h_1 &0 & \cellcolor{grey} \nicefrac{2}{3} & \cellcolor{grey} \nicefrac{1}{3} \\
	3: & h_1 & h_3 & h_2 &\cellcolor{grey} \nicefrac{1}{3} & \nicefrac{1}{3} & \cellcolor{grey} \nicefrac{1}{3} 
\end{array}
$\\\vspace{1em}$
\begin{array}{ccccccc}
	&&&& h_1 & h_2 & h_3 \\\cmidrule{5-7}
	1: & h_1 & h_3 & h_2 &1 & 0 & 0 \\
	2: & h_3 & h_2 & h_1 &0 & 0 & 1 \\
	3: & h_2 & h_1 & h_3 &0 & 1 & 0 
\end{array}
$\hspace{2em}$
\begin{array}{ccccccc}
	&&&& h_1 & h_2 & h_3 \\\cmidrule{5-7}
	1: & h_1 & h_3 & h_2 &1 & 0 & 0 \\
	2: & h_3 & h_2 & h_1 &0 & 0 & 1 \\
	3: & h_2 & h_3 & h_1 &0 & 1 & 0 
\end{array}
$\hspace{2em}$
\begin{array}{ccccccc}
	&&&& h_1 & h_2 & h_3 \\\cmidrule{5-7}
	1: & h_1 & h_3 & h_2 &\cellcolor{grey} 1 & \cellcolor{grey} 0 & 0 \\
	2: & h_3 & h_2 & h_1 &0 & \cellcolor{grey} \nicefrac{2}{3} & \cellcolor{grey} \nicefrac{1}{3} \\
	3: & h_3 & h_1 & h_2 &\cellcolor{grey} 0 & \nicefrac{1}{3} & \cellcolor{grey} \nicefrac{2}{3} 
\end{array}
$\hspace{2em}$
\begin{array}{ccccccc}
	&&&& h_1 & h_2 & h_3 \\\cmidrule{5-7}
	1: & h_1 & h_3 & h_2 &1 & 0 & 0 \\
	2: & h_3 & h_2 & h_1 &0 & \cellcolor{grey} \nicefrac{2}{3} & \cellcolor{grey} \nicefrac{1}{3} \\
	3: & h_3 & h_2 & h_1 &0 & \cellcolor{grey} \nicefrac{1}{3} & \cellcolor{grey} \nicefrac{2}{3} 
\end{array}
$\\\vspace{1em}$
\begin{array}{ccccccc}
	&&&& h_1 & h_2 & h_3 \\\cmidrule{5-7}
	1: & h_2 & h_1 & h_3 &0 & \nicefrac{2}{3} & \nicefrac{1}{3} \\
	2: & h_1 & h_2 & h_3 &\cellcolor{grey} \nicefrac{1}{3} & \cellcolor{grey} \nicefrac{1}{3} & \nicefrac{1}{3} \\
	3: & h_1 & h_2 & h_3 &\cellcolor{grey} \nicefrac{2}{3} & \cellcolor{grey} 0 & \nicefrac{1}{3} 
\end{array}
$\hspace{2em}$
\begin{array}{ccccccc}
	&&&& h_1 & h_2 & h_3 \\\cmidrule{5-7}
	1: & h_2 & h_1 & h_3 &0 & \cellcolor{grey} 1 & \cellcolor{grey} 0 \\
	2: & h_1 & h_2 & h_3 &\cellcolor{grey} \nicefrac{1}{3} & \cellcolor{grey} 0 & \cellcolor{grey} \nicefrac{2}{3} \\
	3: & h_1 & h_3 & h_2 &\cellcolor{grey} \nicefrac{2}{3} & 0 & \cellcolor{grey} \nicefrac{1}{3} 
\end{array}
$\hspace{2em}$
\begin{array}{ccccccc}
	&&&& h_1 & h_2 & h_3 \\\cmidrule{5-7}
	1: & h_2 & h_1 & h_3 &\cellcolor{grey} \nicefrac{1}{3} & \cellcolor{grey} \nicefrac{1}{3} & \nicefrac{1}{3} \\
	2: & h_1 & h_2 & h_3 &\nicefrac{2}{3} & 0 & \nicefrac{1}{3} \\
	3: & h_2 & h_1 & h_3 &\cellcolor{grey} 0 & \cellcolor{grey} \nicefrac{2}{3} & \nicefrac{1}{3} 
\end{array}
$\hspace{2em}$
\begin{array}{ccccccc}
	&&&& h_1 & h_2 & h_3 \\\cmidrule{5-7}
	1: & h_2 & h_1 & h_3 &\cellcolor{grey} 0 & \cellcolor{grey} \nicefrac{1}{3} & \cellcolor{grey} \nicefrac{2}{3} \\
	2: & h_1 & h_2 & h_3 &\cellcolor{grey} 1 & 0 & \cellcolor{grey} 0 \\
	3: & h_2 & h_3 & h_1 &0 & \cellcolor{grey} \nicefrac{2}{3} & \cellcolor{grey} \nicefrac{1}{3} 
\end{array}
$\\\vspace{1em}$
\begin{array}{ccccccc}
	&&&& h_1 & h_2 & h_3 \\\cmidrule{5-7}
	1: & h_2 & h_1 & h_3 &0 & 1 & 0 \\
	2: & h_1 & h_2 & h_3 &1 & 0 & 0 \\
	3: & h_3 & h_1 & h_2 &0 & 0 & 1 
\end{array}
$\hspace{2em}$
\begin{array}{ccccccc}
	&&&& h_1 & h_2 & h_3 \\\cmidrule{5-7}
	1: & h_2 & h_1 & h_3 &0 & 1 & 0 \\
	2: & h_1 & h_2 & h_3 &1 & 0 & 0 \\
	3: & h_3 & h_2 & h_1 &0 & 0 & 1 
\end{array}
$\hspace{2em}$
\begin{array}{ccccccc}
	&&&& h_1 & h_2 & h_3 \\\cmidrule{5-7}
	1: & h_2 & h_1 & h_3 &0 & \cellcolor{grey} 1 & \cellcolor{grey} 0 \\
	2: & h_1 & h_3 & h_2 &\cellcolor{grey} \nicefrac{1}{3} & 0 & \cellcolor{grey} \nicefrac{2}{3} \\
	3: & h_1 & h_2 & h_3 &\cellcolor{grey} \nicefrac{2}{3} & \cellcolor{grey} 0 & \nicefrac{1}{3} 
\end{array}
$\hspace{2em}$
\begin{array}{ccccccc}
	&&&& h_1 & h_2 & h_3 \\\cmidrule{5-7}
	1: & h_2 & h_1 & h_3 &0 & 1 & 0 \\
	2: & h_1 & h_3 & h_2 &\cellcolor{grey} \nicefrac{1}{3} & 0 & \cellcolor{grey} \nicefrac{2}{3} \\
	3: & h_1 & h_3 & h_2 &\cellcolor{grey} \nicefrac{2}{3} & 0 & \cellcolor{grey} \nicefrac{1}{3} 
\end{array}
$\\\vspace{1em}$
\begin{array}{ccccccc}
	&&&& h_1 & h_2 & h_3 \\\cmidrule{5-7}
	1: & h_2 & h_1 & h_3 &\cellcolor{grey} 0 & \cellcolor{grey} \nicefrac{2}{3} & \nicefrac{1}{3} \\
	2: & h_1 & h_3 & h_2 &\nicefrac{2}{3} & 0 & \nicefrac{1}{3} \\
	3: & h_2 & h_1 & h_3 &\cellcolor{grey} \nicefrac{1}{3} & \cellcolor{grey} \nicefrac{1}{3} & \nicefrac{1}{3} 
\end{array}
$\hspace{2em}$
\begin{array}{ccccccc}
	&&&& h_1 & h_2 & h_3 \\\cmidrule{5-7}
	1: & h_2 & h_1 & h_3 &\cellcolor{grey} 0 & \cellcolor{grey} \nicefrac{2}{3} & \nicefrac{1}{3} \\
	2: & h_1 & h_3 & h_2 &\cellcolor{grey} 1 & 0 & \cellcolor{grey} 0 \\
	3: & h_2 & h_3 & h_1 &0 & \cellcolor{grey} \nicefrac{1}{3} & \cellcolor{grey} \nicefrac{2}{3} 
\end{array}
$\hspace{2em}$
\begin{array}{ccccccc}
	&&&& h_1 & h_2 & h_3 \\\cmidrule{5-7}
	1: & h_2 & h_1 & h_3 &0 & 1 & 0 \\
	2: & h_1 & h_3 & h_2 &1 & 0 & 0 \\
	3: & h_3 & h_1 & h_2 &0 & 0 & 1 
\end{array}
$\hspace{2em}$
\begin{array}{ccccccc}
	&&&& h_1 & h_2 & h_3 \\\cmidrule{5-7}
	1: & h_2 & h_1 & h_3 &0 & 1 & 0 \\
	2: & h_1 & h_3 & h_2 &1 & 0 & 0 \\
	3: & h_3 & h_2 & h_1 &0 & 0 & 1 
\end{array}
$\\\vspace{1em}$
\begin{array}{ccccccc}
	&&&& h_1 & h_2 & h_3 \\\cmidrule{5-7}
	1: & h_2 & h_1 & h_3 &\cellcolor{grey} 0 & \cellcolor{grey} \nicefrac{2}{3} & \nicefrac{1}{3} \\
	2: & h_2 & h_1 & h_3 &\cellcolor{grey} \nicefrac{1}{3} & \cellcolor{grey} \nicefrac{1}{3} & \nicefrac{1}{3} \\
	3: & h_1 & h_2 & h_3 &\nicefrac{2}{3} & 0 & \nicefrac{1}{3} 
\end{array}
$\hspace{2em}$
\begin{array}{ccccccc}
	&&&& h_1 & h_2 & h_3 \\\cmidrule{5-7}
	1: & h_2 & h_1 & h_3 &\cellcolor{grey} \nicefrac{1}{3} & \cellcolor{grey} \nicefrac{2}{3} & \cellcolor{grey} 0 \\
	2: & h_2 & h_1 & h_3 &\cellcolor{grey} 0 & \cellcolor{grey} \nicefrac{1}{3} & \cellcolor{grey} \nicefrac{2}{3} \\
	3: & h_1 & h_3 & h_2 &\nicefrac{2}{3} & 0 & \nicefrac{1}{3} 
\end{array}
$\hspace{2em}$
\begin{array}{ccccccc}
	&&&& h_1 & h_2 & h_3 \\\cmidrule{5-7}
	1: & h_2 & h_1 & h_3 &\nicefrac{1}{3} & \nicefrac{1}{3} & \nicefrac{1}{3} \\
	2: & h_2 & h_1 & h_3 &\nicefrac{1}{3} & \nicefrac{1}{3} & \nicefrac{1}{3} \\
	3: & h_2 & h_1 & h_3 &\nicefrac{1}{3} & \nicefrac{1}{3} & \nicefrac{1}{3} 
\end{array}
$\hspace{2em}$
\begin{array}{ccccccc}
	&&&& h_1 & h_2 & h_3 \\\cmidrule{5-7}
	1: & h_2 & h_1 & h_3 &\cellcolor{grey} \nicefrac{2}{3} & \cellcolor{grey} 0 & \cellcolor{grey} \nicefrac{1}{3} \\
	2: & h_2 & h_1 & h_3 &\cellcolor{grey} \nicefrac{1}{3} & \cellcolor{grey} \nicefrac{2}{3} & \cellcolor{grey} 0 \\
	3: & h_2 & h_3 & h_1 &0 & \nicefrac{1}{3} & \nicefrac{2}{3} 
\end{array}
$\\\vspace{1em}$
\begin{array}{ccccccc}
	&&&& h_1 & h_2 & h_3 \\\cmidrule{5-7}
	1: & h_2 & h_1 & h_3 &\cellcolor{grey} 1 & \cellcolor{grey} 0 & 0 \\
	2: & h_2 & h_1 & h_3 &\cellcolor{grey} 0 & \cellcolor{grey} 1 & 0 \\
	3: & h_3 & h_1 & h_2 &0 & 0 & 1 
\end{array}
$\hspace{2em}$
\begin{array}{ccccccc}
	&&&& h_1 & h_2 & h_3 \\\cmidrule{5-7}
	1: & h_2 & h_1 & h_3 &\cellcolor{grey} 1 & \cellcolor{grey} 0 & 0 \\
	2: & h_2 & h_1 & h_3 &\cellcolor{grey} 0 & \cellcolor{grey} 1 & 0 \\
	3: & h_3 & h_2 & h_1 &0 & 0 & 1 
\end{array}
$\hspace{2em}$
\begin{array}{ccccccc}
	&&&& h_1 & h_2 & h_3 \\\cmidrule{5-7}
	1: & h_2 & h_1 & h_3 &\cellcolor{grey} 0 & \cellcolor{grey} \nicefrac{2}{3} & \nicefrac{1}{3} \\
	2: & h_2 & h_3 & h_1 &0 & \cellcolor{grey} \nicefrac{1}{3} & \cellcolor{grey} \nicefrac{2}{3} \\
	3: & h_1 & h_2 & h_3 &\cellcolor{grey} 1 & 0 & \cellcolor{grey} 0 
\end{array}
$\hspace{2em}$
\begin{array}{ccccccc}
	&&&& h_1 & h_2 & h_3 \\\cmidrule{5-7}
	1: & h_2 & h_1 & h_3 &\cellcolor{grey} 0 & \cellcolor{grey} \nicefrac{2}{3} & \nicefrac{1}{3} \\
	2: & h_2 & h_3 & h_1 &0 & \cellcolor{grey} \nicefrac{1}{3} & \cellcolor{grey} \nicefrac{2}{3} \\
	3: & h_1 & h_3 & h_2 &\cellcolor{grey} 1 & 0 & \cellcolor{grey} 0 
\end{array}
$\\\vspace{1em}$
\begin{array}{ccccccc}
	&&&& h_1 & h_2 & h_3 \\\cmidrule{5-7}
	1: & h_2 & h_1 & h_3 &\cellcolor{grey} 0 & \cellcolor{grey} \nicefrac{2}{3} & \cellcolor{grey} \nicefrac{1}{3} \\
	2: & h_2 & h_3 & h_1 &0 & \nicefrac{1}{3} & \nicefrac{2}{3} \\
	3: & h_2 & h_1 & h_3 &\cellcolor{grey} 1 & \cellcolor{grey} 0 & \cellcolor{grey} 0 
\end{array}
$\hspace{2em}$
\begin{array}{ccccccc}
	&&&& h_1 & h_2 & h_3 \\\cmidrule{5-7}
	1: & h_2 & h_1 & h_3 &\nicefrac{2}{3} & \nicefrac{1}{3} & 0 \\
	2: & h_2 & h_3 & h_1 &\cellcolor{grey} \nicefrac{1}{3} & \cellcolor{grey} \nicefrac{2}{3} & \cellcolor{grey} 0 \\
	3: & h_2 & h_3 & h_1 &\cellcolor{grey} 0 & \cellcolor{grey} 0 & \cellcolor{grey} 1 
\end{array}
$\hspace{2em}$
\begin{array}{ccccccc}
	&&&& h_1 & h_2 & h_3 \\\cmidrule{5-7}
	1: & h_2 & h_1 & h_3 &\cellcolor{grey} 1 & \cellcolor{grey} 0 & 0 \\
	2: & h_2 & h_3 & h_1 &\cellcolor{grey} 0 & \cellcolor{grey} 1 & \cellcolor{grey} 0 \\
	3: & h_3 & h_1 & h_2 &\cellcolor{grey} 0 & 0 & \cellcolor{grey} 1 
\end{array}
$\hspace{2em}$
\begin{array}{ccccccc}
	&&&& h_1 & h_2 & h_3 \\\cmidrule{5-7}
	1: & h_2 & h_1 & h_3 &\cellcolor{grey} 1 & \cellcolor{grey} 0 & 0 \\
	2: & h_2 & h_3 & h_1 &\cellcolor{grey} 0 & \cellcolor{grey} 1 & \cellcolor{grey} 0 \\
	3: & h_3 & h_2 & h_1 &\cellcolor{grey} 0 & 0 & \cellcolor{grey} 1 
\end{array}
$\\\vspace{1em}$
\begin{array}{ccccccc}
	&&&& h_1 & h_2 & h_3 \\\cmidrule{5-7}
	1: & h_2 & h_1 & h_3 &0 & 1 & 0 \\
	2: & h_3 & h_1 & h_2 &0 & 0 & 1 \\
	3: & h_1 & h_2 & h_3 &1 & 0 & 0 
\end{array}
$\hspace{2em}$
\begin{array}{ccccccc}
	&&&& h_1 & h_2 & h_3 \\\cmidrule{5-7}
	1: & h_2 & h_1 & h_3 &0 & 1 & 0 \\
	2: & h_3 & h_1 & h_2 &0 & 0 & 1 \\
	3: & h_1 & h_3 & h_2 &1 & 0 & 0 
\end{array}
$\hspace{2em}$
\begin{array}{ccccccc}
	&&&& h_1 & h_2 & h_3 \\\cmidrule{5-7}
	1: & h_2 & h_1 & h_3 &\cellcolor{grey} \nicefrac{2}{3} & \cellcolor{grey} \nicefrac{1}{3} & 0 \\
	2: & h_3 & h_1 & h_2 &0 & 0 & 1 \\
	3: & h_2 & h_1 & h_3 &\cellcolor{grey} \nicefrac{1}{3} & \cellcolor{grey} \nicefrac{2}{3} & 0 
\end{array}
$\hspace{2em}$
\begin{array}{ccccccc}
	&&&& h_1 & h_2 & h_3 \\\cmidrule{5-7}
	1: & h_2 & h_1 & h_3 &\cellcolor{grey} \nicefrac{2}{3} & \cellcolor{grey} \nicefrac{1}{3} & 0 \\
	2: & h_3 & h_1 & h_2 &\cellcolor{grey} \nicefrac{1}{3} & 0 & \cellcolor{grey} \nicefrac{2}{3} \\
	3: & h_2 & h_3 & h_1 &\cellcolor{grey} 0 & \cellcolor{grey} \nicefrac{2}{3} & \cellcolor{grey} \nicefrac{1}{3} 
\end{array}
$\\\vspace{1em}$
\begin{array}{ccccccc}
	&&&& h_1 & h_2 & h_3 \\\cmidrule{5-7}
	1: & h_2 & h_1 & h_3 &0 & 1 & 0 \\
	2: & h_3 & h_1 & h_2 &\cellcolor{grey} 1 & 0 & \cellcolor{grey} 0 \\
	3: & h_3 & h_1 & h_2 &\cellcolor{grey} 0 & 0 & \cellcolor{grey} 1 
\end{array}
$\hspace{2em}$
\begin{array}{ccccccc}
	&&&& h_1 & h_2 & h_3 \\\cmidrule{5-7}
	1: & h_2 & h_1 & h_3 &\cellcolor{grey} 0 & \cellcolor{grey} 1 & 0 \\
	2: & h_3 & h_1 & h_2 &\cellcolor{grey} 1 & 0 & \cellcolor{grey} 0 \\
	3: & h_3 & h_2 & h_1 &\cellcolor{grey} 0 & \cellcolor{grey} 0 & \cellcolor{grey} 1 
\end{array}
$\hspace{2em}$
\begin{array}{ccccccc}
	&&&& h_1 & h_2 & h_3 \\\cmidrule{5-7}
	1: & h_2 & h_1 & h_3 &0 & 1 & 0 \\
	2: & h_3 & h_2 & h_1 &0 & 0 & 1 \\
	3: & h_1 & h_2 & h_3 &1 & 0 & 0 
\end{array}
$\hspace{2em}$
\begin{array}{ccccccc}
	&&&& h_1 & h_2 & h_3 \\\cmidrule{5-7}
	1: & h_2 & h_1 & h_3 &0 & 1 & 0 \\
	2: & h_3 & h_2 & h_1 &0 & 0 & 1 \\
	3: & h_1 & h_3 & h_2 &1 & 0 & 0 
\end{array}
$\\\vspace{1em}$
\begin{array}{ccccccc}
	&&&& h_1 & h_2 & h_3 \\\cmidrule{5-7}
	1: & h_2 & h_1 & h_3 &\cellcolor{grey} \nicefrac{2}{3} & \cellcolor{grey} \nicefrac{1}{3} & 0 \\
	2: & h_3 & h_2 & h_1 &0 & 0 & 1 \\
	3: & h_2 & h_1 & h_3 &\cellcolor{grey} \nicefrac{1}{3} & \cellcolor{grey} \nicefrac{2}{3} & 0 
\end{array}
$\hspace{2em}$
\begin{array}{ccccccc}
	&&&& h_1 & h_2 & h_3 \\\cmidrule{5-7}
	1: & h_2 & h_1 & h_3 &\cellcolor{grey} \nicefrac{2}{3} & \cellcolor{grey} \nicefrac{1}{3} & 0 \\
	2: & h_3 & h_2 & h_1 &\cellcolor{grey} \nicefrac{1}{3} & 0 & \cellcolor{grey} \nicefrac{2}{3} \\
	3: & h_2 & h_3 & h_1 &\cellcolor{grey} 0 & \cellcolor{grey} \nicefrac{2}{3} & \cellcolor{grey} \nicefrac{1}{3} 
\end{array}
$\hspace{2em}$
\begin{array}{ccccccc}
	&&&& h_1 & h_2 & h_3 \\\cmidrule{5-7}
	1: & h_2 & h_1 & h_3 &\cellcolor{grey} 1 & \cellcolor{grey} 0 & 0 \\
	2: & h_3 & h_2 & h_1 &\cellcolor{grey} 0 & \cellcolor{grey} 1 & \cellcolor{grey} 0 \\
	3: & h_3 & h_1 & h_2 &\cellcolor{grey} 0 & 0 & \cellcolor{grey} 1 
\end{array}
$\hspace{2em}$
\begin{array}{ccccccc}
	&&&& h_1 & h_2 & h_3 \\\cmidrule{5-7}
	1: & h_2 & h_1 & h_3 &\cellcolor{grey} 1 & \cellcolor{grey} 0 & 0 \\
	2: & h_3 & h_2 & h_1 &\cellcolor{grey} 0 & \cellcolor{grey} 1 & \cellcolor{grey} 0 \\
	3: & h_3 & h_2 & h_1 &\cellcolor{grey} 0 & \cellcolor{grey} 0 & \cellcolor{grey} 1 
\end{array}
$\\\vspace{1em}$
\begin{array}{ccccccc}
	&&&& h_1 & h_2 & h_3 \\\cmidrule{5-7}
	1: & h_2 & h_3 & h_1 &0 & \nicefrac{2}{3} & \nicefrac{1}{3} \\
	2: & h_1 & h_2 & h_3 &\cellcolor{grey} \nicefrac{2}{3} & \cellcolor{grey} 0 & \nicefrac{1}{3} \\
	3: & h_1 & h_2 & h_3 &\cellcolor{grey} \nicefrac{1}{3} & \cellcolor{grey} \nicefrac{1}{3} & \nicefrac{1}{3} 
\end{array}
$\hspace{2em}$
\begin{array}{ccccccc}
	&&&& h_1 & h_2 & h_3 \\\cmidrule{5-7}
	1: & h_2 & h_3 & h_1 &0 & \cellcolor{grey} 1 & \cellcolor{grey} 0 \\
	2: & h_1 & h_2 & h_3 &\cellcolor{grey} \nicefrac{2}{3} & \cellcolor{grey} 0 & \nicefrac{1}{3} \\
	3: & h_1 & h_3 & h_2 &\cellcolor{grey} \nicefrac{1}{3} & 0 & \cellcolor{grey} \nicefrac{2}{3} 
\end{array}
$\hspace{2em}$
\begin{array}{ccccccc}
	&&&& h_1 & h_2 & h_3 \\\cmidrule{5-7}
	1: & h_2 & h_3 & h_1 &0 & \cellcolor{grey} \nicefrac{1}{3} & \cellcolor{grey} \nicefrac{2}{3} \\
	2: & h_1 & h_2 & h_3 &\cellcolor{grey} 1 & 0 & \cellcolor{grey} 0 \\
	3: & h_2 & h_1 & h_3 &\cellcolor{grey} 0 & \cellcolor{grey} \nicefrac{2}{3} & \nicefrac{1}{3} 
\end{array}
$\hspace{2em}$
\begin{array}{ccccccc}
	&&&& h_1 & h_2 & h_3 \\\cmidrule{5-7}
	1: & h_2 & h_3 & h_1 &0 & \cellcolor{grey} \nicefrac{1}{3} & \cellcolor{grey} \nicefrac{2}{3} \\
	2: & h_1 & h_2 & h_3 &1 & 0 & 0 \\
	3: & h_2 & h_3 & h_1 &0 & \cellcolor{grey} \nicefrac{2}{3} & \cellcolor{grey} \nicefrac{1}{3} 
\end{array}
$\\\vspace{1em}$
\begin{array}{ccccccc}
	&&&& h_1 & h_2 & h_3 \\\cmidrule{5-7}
	1: & h_2 & h_3 & h_1 &0 & 1 & 0 \\
	2: & h_1 & h_2 & h_3 &1 & 0 & 0 \\
	3: & h_3 & h_1 & h_2 &0 & 0 & 1 
\end{array}
$\hspace{2em}$
\begin{array}{ccccccc}
	&&&& h_1 & h_2 & h_3 \\\cmidrule{5-7}
	1: & h_2 & h_3 & h_1 &0 & 1 & 0 \\
	2: & h_1 & h_2 & h_3 &1 & 0 & 0 \\
	3: & h_3 & h_2 & h_1 &0 & 0 & 1 
\end{array}
$\hspace{2em}$
\begin{array}{ccccccc}
	&&&& h_1 & h_2 & h_3 \\\cmidrule{5-7}
	1: & h_2 & h_3 & h_1 &0 & \cellcolor{grey} 1 & \cellcolor{grey} 0 \\
	2: & h_1 & h_3 & h_2 &\cellcolor{grey} \nicefrac{2}{3} & 0 & \cellcolor{grey} \nicefrac{1}{3} \\
	3: & h_1 & h_2 & h_3 &\cellcolor{grey} \nicefrac{1}{3} & \cellcolor{grey} 0 & \cellcolor{grey} \nicefrac{2}{3} 
\end{array}
$\hspace{2em}$
\begin{array}{ccccccc}
	&&&& h_1 & h_2 & h_3 \\\cmidrule{5-7}
	1: & h_2 & h_3 & h_1 &0 & 1 & 0 \\
	2: & h_1 & h_3 & h_2 &\cellcolor{grey} \nicefrac{2}{3} & 0 & \cellcolor{grey} \nicefrac{1}{3} \\
	3: & h_1 & h_3 & h_2 &\cellcolor{grey} \nicefrac{1}{3} & 0 & \cellcolor{grey} \nicefrac{2}{3} 
\end{array}
$\\\vspace{1em}$
\begin{array}{ccccccc}
	&&&& h_1 & h_2 & h_3 \\\cmidrule{5-7}
	1: & h_2 & h_3 & h_1 &0 & \cellcolor{grey} \nicefrac{2}{3} & \cellcolor{grey} \nicefrac{1}{3} \\
	2: & h_1 & h_3 & h_2 &\cellcolor{grey} 1 & 0 & \cellcolor{grey} 0 \\
	3: & h_2 & h_1 & h_3 &\cellcolor{grey} 0 & \cellcolor{grey} \nicefrac{1}{3} & \cellcolor{grey} \nicefrac{2}{3} 
\end{array}
$\hspace{2em}$
\begin{array}{ccccccc}
	&&&& h_1 & h_2 & h_3 \\\cmidrule{5-7}
	1: & h_2 & h_3 & h_1 &0 & \cellcolor{grey} \nicefrac{2}{3} & \cellcolor{grey} \nicefrac{1}{3} \\
	2: & h_1 & h_3 & h_2 &1 & 0 & 0 \\
	3: & h_2 & h_3 & h_1 &0 & \cellcolor{grey} \nicefrac{1}{3} & \cellcolor{grey} \nicefrac{2}{3} 
\end{array}
$\hspace{2em}$
\begin{array}{ccccccc}
	&&&& h_1 & h_2 & h_3 \\\cmidrule{5-7}
	1: & h_2 & h_3 & h_1 &0 & 1 & 0 \\
	2: & h_1 & h_3 & h_2 &1 & 0 & 0 \\
	3: & h_3 & h_1 & h_2 &0 & 0 & 1 
\end{array}
$\hspace{2em}$
\begin{array}{ccccccc}
	&&&& h_1 & h_2 & h_3 \\\cmidrule{5-7}
	1: & h_2 & h_3 & h_1 &0 & 1 & 0 \\
	2: & h_1 & h_3 & h_2 &1 & 0 & 0 \\
	3: & h_3 & h_2 & h_1 &0 & 0 & 1 
\end{array}
$\\\vspace{1em}$
\begin{array}{ccccccc}
	&&&& h_1 & h_2 & h_3 \\\cmidrule{5-7}
	1: & h_2 & h_3 & h_1 &0 & \cellcolor{grey} \nicefrac{2}{3} & \cellcolor{grey} \nicefrac{1}{3} \\
	2: & h_2 & h_1 & h_3 &\cellcolor{grey} \nicefrac{1}{3} & \cellcolor{grey} \nicefrac{1}{3} & \nicefrac{1}{3} \\
	3: & h_1 & h_2 & h_3 &\cellcolor{grey} \nicefrac{2}{3} & 0 & \cellcolor{grey} \nicefrac{1}{3} 
\end{array}
$\hspace{2em}$
\begin{array}{ccccccc}
	&&&& h_1 & h_2 & h_3 \\\cmidrule{5-7}
	1: & h_2 & h_3 & h_1 &0 & \cellcolor{grey} \nicefrac{2}{3} & \cellcolor{grey} \nicefrac{1}{3} \\
	2: & h_2 & h_1 & h_3 &\cellcolor{grey} \nicefrac{1}{3} & \cellcolor{grey} \nicefrac{1}{3} & \nicefrac{1}{3} \\
	3: & h_1 & h_3 & h_2 &\cellcolor{grey} \nicefrac{2}{3} & 0 & \cellcolor{grey} \nicefrac{1}{3} 
\end{array}
$\hspace{2em}$
\begin{array}{ccccccc}
	&&&& h_1 & h_2 & h_3 \\\cmidrule{5-7}
	1: & h_2 & h_3 & h_1 &0 & \nicefrac{1}{3} & \nicefrac{2}{3} \\
	2: & h_2 & h_1 & h_3 &\cellcolor{grey} 1 & \cellcolor{grey} 0 & \cellcolor{grey} 0 \\
	3: & h_2 & h_1 & h_3 &\cellcolor{grey} 0 & \cellcolor{grey} \nicefrac{2}{3} & \cellcolor{grey} \nicefrac{1}{3} 
\end{array}
$\hspace{2em}$
\begin{array}{ccccccc}
	&&&& h_1 & h_2 & h_3 \\\cmidrule{5-7}
	1: & h_2 & h_3 & h_1 &\cellcolor{grey} \nicefrac{1}{3} & \cellcolor{grey} 0 & \cellcolor{grey} \nicefrac{2}{3} \\
	2: & h_2 & h_1 & h_3 &\nicefrac{2}{3} & \nicefrac{1}{3} & 0 \\
	3: & h_2 & h_3 & h_1 &\cellcolor{grey} 0 & \cellcolor{grey} \nicefrac{2}{3} & \cellcolor{grey} \nicefrac{1}{3} 
\end{array}
$\\\vspace{1em}$
\begin{array}{ccccccc}
	&&&& h_1 & h_2 & h_3 \\\cmidrule{5-7}
	1: & h_2 & h_3 & h_1 &\cellcolor{grey} 1 & \cellcolor{grey} 0 & \cellcolor{grey} 0 \\
	2: & h_2 & h_1 & h_3 &\cellcolor{grey} 0 & \cellcolor{grey} 1 & 0 \\
	3: & h_3 & h_1 & h_2 &\cellcolor{grey} 0 & 0 & \cellcolor{grey} 1 
\end{array}
$\hspace{2em}$
\begin{array}{ccccccc}
	&&&& h_1 & h_2 & h_3 \\\cmidrule{5-7}
	1: & h_2 & h_3 & h_1 &\cellcolor{grey} 1 & \cellcolor{grey} 0 & \cellcolor{grey} 0 \\
	2: & h_2 & h_1 & h_3 &\cellcolor{grey} 0 & \cellcolor{grey} 1 & 0 \\
	3: & h_3 & h_2 & h_1 &\cellcolor{grey} 0 & 0 & \cellcolor{grey} 1 
\end{array}
$\hspace{2em}$
\begin{array}{ccccccc}
	&&&& h_1 & h_2 & h_3 \\\cmidrule{5-7}
	1: & h_2 & h_3 & h_1 &0 & \cellcolor{grey} \nicefrac{2}{3} & \cellcolor{grey} \nicefrac{1}{3} \\
	2: & h_2 & h_3 & h_1 &0 & \cellcolor{grey} \nicefrac{1}{3} & \cellcolor{grey} \nicefrac{2}{3} \\
	3: & h_1 & h_2 & h_3 &1 & 0 & 0 
\end{array}
$\hspace{2em}$
\begin{array}{ccccccc}
	&&&& h_1 & h_2 & h_3 \\\cmidrule{5-7}
	1: & h_2 & h_3 & h_1 &0 & \cellcolor{grey} \nicefrac{2}{3} & \cellcolor{grey} \nicefrac{1}{3} \\
	2: & h_2 & h_3 & h_1 &0 & \cellcolor{grey} \nicefrac{1}{3} & \cellcolor{grey} \nicefrac{2}{3} \\
	3: & h_1 & h_3 & h_2 &1 & 0 & 0 
\end{array}
$\\\vspace{1em}$
\begin{array}{ccccccc}
	&&&& h_1 & h_2 & h_3 \\\cmidrule{5-7}
	1: & h_2 & h_3 & h_1 &\cellcolor{grey} 0 & \cellcolor{grey} \nicefrac{2}{3} & \cellcolor{grey} \nicefrac{1}{3} \\
	2: & h_2 & h_3 & h_1 &\cellcolor{grey} \nicefrac{1}{3} & \cellcolor{grey} 0 & \cellcolor{grey} \nicefrac{2}{3} \\
	3: & h_2 & h_1 & h_3 &\nicefrac{2}{3} & \nicefrac{1}{3} & 0 
\end{array}
$\hspace{2em}$
\begin{array}{ccccccc}
	&&&& h_1 & h_2 & h_3 \\\cmidrule{5-7}
	1: & h_2 & h_3 & h_1 &\cellcolor{grey} \nicefrac{2}{3} & \nicefrac{1}{3} & \cellcolor{grey} 0 \\
	2: & h_2 & h_3 & h_1 &\nicefrac{1}{3} & \nicefrac{1}{3} & \nicefrac{1}{3} \\
	3: & h_2 & h_3 & h_1 &\cellcolor{grey} 0 & \nicefrac{1}{3} & \cellcolor{grey} \nicefrac{2}{3} 
\end{array}
$\hspace{2em}$
\begin{array}{ccccccc}
	&&&& h_1 & h_2 & h_3 \\\cmidrule{5-7}
	1: & h_2 & h_3 & h_1 &\cellcolor{grey} 1 & \cellcolor{grey} 0 & \cellcolor{grey} 0 \\
	2: & h_2 & h_3 & h_1 &\cellcolor{grey} 0 & \cellcolor{grey} 1 & \cellcolor{grey} 0 \\
	3: & h_3 & h_1 & h_2 &\cellcolor{grey} 0 & 0 & \cellcolor{grey} 1 
\end{array}
$\hspace{2em}$
\begin{array}{ccccccc}
	&&&& h_1 & h_2 & h_3 \\\cmidrule{5-7}
	1: & h_2 & h_3 & h_1 &\cellcolor{grey} 1 & \cellcolor{grey} 0 & \cellcolor{grey} 0 \\
	2: & h_2 & h_3 & h_1 &\cellcolor{grey} 0 & \cellcolor{grey} 1 & \cellcolor{grey} 0 \\
	3: & h_3 & h_2 & h_1 &\cellcolor{grey} 0 & 0 & \cellcolor{grey} 1 
\end{array}
$\\\vspace{1em}$
\begin{array}{ccccccc}
	&&&& h_1 & h_2 & h_3 \\\cmidrule{5-7}
	1: & h_2 & h_3 & h_1 &0 & 1 & 0 \\
	2: & h_3 & h_1 & h_2 &0 & 0 & 1 \\
	3: & h_1 & h_2 & h_3 &1 & 0 & 0 
\end{array}
$\hspace{2em}$
\begin{array}{ccccccc}
	&&&& h_1 & h_2 & h_3 \\\cmidrule{5-7}
	1: & h_2 & h_3 & h_1 &0 & 1 & 0 \\
	2: & h_3 & h_1 & h_2 &0 & 0 & 1 \\
	3: & h_1 & h_3 & h_2 &1 & 0 & 0 
\end{array}
$\hspace{2em}$
\begin{array}{ccccccc}
	&&&& h_1 & h_2 & h_3 \\\cmidrule{5-7}
	1: & h_2 & h_3 & h_1 &\nicefrac{1}{3} & \cellcolor{grey} \nicefrac{1}{3} & \cellcolor{grey} \nicefrac{1}{3} \\
	2: & h_3 & h_1 & h_2 &\cellcolor{grey} \nicefrac{1}{3} & 0 & \cellcolor{grey} \nicefrac{2}{3} \\
	3: & h_2 & h_1 & h_3 &\cellcolor{grey} \nicefrac{1}{3} & \cellcolor{grey} \nicefrac{2}{3} & 0 
\end{array}
$\hspace{2em}$
\begin{array}{ccccccc}
	&&&& h_1 & h_2 & h_3 \\\cmidrule{5-7}
	1: & h_2 & h_3 & h_1 &\cellcolor{grey} \nicefrac{2}{3} & \cellcolor{grey} \nicefrac{1}{3} & \cellcolor{grey} 0 \\
	2: & h_3 & h_1 & h_2 &\nicefrac{1}{3} & 0 & \nicefrac{2}{3} \\
	3: & h_2 & h_3 & h_1 &\cellcolor{grey} 0 & \cellcolor{grey} \nicefrac{2}{3} & \cellcolor{grey} \nicefrac{1}{3} 
\end{array}
$\\\vspace{1em}$
\begin{array}{ccccccc}
	&&&& h_1 & h_2 & h_3 \\\cmidrule{5-7}
	1: & h_2 & h_3 & h_1 &0 & 1 & 0 \\
	2: & h_3 & h_1 & h_2 &\cellcolor{grey} 1 & 0 & \cellcolor{grey} 0 \\
	3: & h_3 & h_1 & h_2 &\cellcolor{grey} 0 & 0 & \cellcolor{grey} 1 
\end{array}
$\hspace{2em}$
\begin{array}{ccccccc}
	&&&& h_1 & h_2 & h_3 \\\cmidrule{5-7}
	1: & h_2 & h_3 & h_1 &\cellcolor{grey} 0 & \cellcolor{grey} 1 & 0 \\
	2: & h_3 & h_1 & h_2 &\cellcolor{grey} 1 & 0 & \cellcolor{grey} 0 \\
	3: & h_3 & h_2 & h_1 &\cellcolor{grey} 0 & \cellcolor{grey} 0 & \cellcolor{grey} 1 
\end{array}
$\hspace{2em}$
\begin{array}{ccccccc}
	&&&& h_1 & h_2 & h_3 \\\cmidrule{5-7}
	1: & h_2 & h_3 & h_1 &0 & 1 & 0 \\
	2: & h_3 & h_2 & h_1 &0 & 0 & 1 \\
	3: & h_1 & h_2 & h_3 &1 & 0 & 0 
\end{array}
$\hspace{2em}$
\begin{array}{ccccccc}
	&&&& h_1 & h_2 & h_3 \\\cmidrule{5-7}
	1: & h_2 & h_3 & h_1 &0 & 1 & 0 \\
	2: & h_3 & h_2 & h_1 &0 & 0 & 1 \\
	3: & h_1 & h_3 & h_2 &1 & 0 & 0 
\end{array}
$\\\vspace{1em}$
\begin{array}{ccccccc}
	&&&& h_1 & h_2 & h_3 \\\cmidrule{5-7}
	1: & h_2 & h_3 & h_1 &\nicefrac{1}{3} & \cellcolor{grey} \nicefrac{1}{3} & \cellcolor{grey} \nicefrac{1}{3} \\
	2: & h_3 & h_2 & h_1 &\cellcolor{grey} \nicefrac{1}{3} & 0 & \cellcolor{grey} \nicefrac{2}{3} \\
	3: & h_2 & h_1 & h_3 &\cellcolor{grey} \nicefrac{1}{3} & \cellcolor{grey} \nicefrac{2}{3} & 0 
\end{array}
$\hspace{2em}$
\begin{array}{ccccccc}
	&&&& h_1 & h_2 & h_3 \\\cmidrule{5-7}
	1: & h_2 & h_3 & h_1 &\cellcolor{grey} \nicefrac{2}{3} & \cellcolor{grey} \nicefrac{1}{3} & \cellcolor{grey} 0 \\
	2: & h_3 & h_2 & h_1 &\nicefrac{1}{3} & 0 & \nicefrac{2}{3} \\
	3: & h_2 & h_3 & h_1 &\cellcolor{grey} 0 & \cellcolor{grey} \nicefrac{2}{3} & \cellcolor{grey} \nicefrac{1}{3} 
\end{array}
$\hspace{2em}$
\begin{array}{ccccccc}
	&&&& h_1 & h_2 & h_3 \\\cmidrule{5-7}
	1: & h_2 & h_3 & h_1 &\cellcolor{grey} 1 & \cellcolor{grey} 0 & 0 \\
	2: & h_3 & h_2 & h_1 &\cellcolor{grey} 0 & \cellcolor{grey} 1 & \cellcolor{grey} 0 \\
	3: & h_3 & h_1 & h_2 &\cellcolor{grey} 0 & 0 & \cellcolor{grey} 1 
\end{array}
$\hspace{2em}$
\begin{array}{ccccccc}
	&&&& h_1 & h_2 & h_3 \\\cmidrule{5-7}
	1: & h_2 & h_3 & h_1 &\cellcolor{grey} 1 & \cellcolor{grey} 0 & 0 \\
	2: & h_3 & h_2 & h_1 &\cellcolor{grey} 0 & \cellcolor{grey} 1 & \cellcolor{grey} 0 \\
	3: & h_3 & h_2 & h_1 &\cellcolor{grey} 0 & \cellcolor{grey} 0 & \cellcolor{grey} 1 
\end{array}
$\\\vspace{1em}$
\begin{array}{ccccccc}
	&&&& h_1 & h_2 & h_3 \\\cmidrule{5-7}
	1: & h_3 & h_1 & h_2 &0 & 0 & 1 \\
	2: & h_1 & h_2 & h_3 &\cellcolor{grey} \nicefrac{1}{3} & \cellcolor{grey} \nicefrac{2}{3} & 0 \\
	3: & h_1 & h_2 & h_3 &\cellcolor{grey} \nicefrac{2}{3} & \cellcolor{grey} \nicefrac{1}{3} & 0 
\end{array}
$\hspace{2em}$
\begin{array}{ccccccc}
	&&&& h_1 & h_2 & h_3 \\\cmidrule{5-7}
	1: & h_3 & h_1 & h_2 &0 & \cellcolor{grey} 0 & \cellcolor{grey} 1 \\
	2: & h_1 & h_2 & h_3 &\cellcolor{grey} \nicefrac{1}{3} & \cellcolor{grey} \nicefrac{2}{3} & 0 \\
	3: & h_1 & h_3 & h_2 &\cellcolor{grey} \nicefrac{2}{3} & \nicefrac{1}{3} & \cellcolor{grey} 0 
\end{array}
$\hspace{2em}$
\begin{array}{ccccccc}
	&&&& h_1 & h_2 & h_3 \\\cmidrule{5-7}
	1: & h_3 & h_1 & h_2 &0 & 0 & 1 \\
	2: & h_1 & h_2 & h_3 &1 & 0 & 0 \\
	3: & h_2 & h_1 & h_3 &0 & 1 & 0 
\end{array}
$\hspace{2em}$
\begin{array}{ccccccc}
	&&&& h_1 & h_2 & h_3 \\\cmidrule{5-7}
	1: & h_3 & h_1 & h_2 &0 & 0 & 1 \\
	2: & h_1 & h_2 & h_3 &1 & 0 & 0 \\
	3: & h_2 & h_3 & h_1 &0 & 1 & 0 
\end{array}
$\\\vspace{1em}$
\begin{array}{ccccccc}
	&&&& h_1 & h_2 & h_3 \\\cmidrule{5-7}
	1: & h_3 & h_1 & h_2 &\cellcolor{grey} 0 & \cellcolor{grey} \nicefrac{2}{3} & \cellcolor{grey} \nicefrac{1}{3} \\
	2: & h_1 & h_2 & h_3 &\cellcolor{grey} 1 & \cellcolor{grey} 0 & 0 \\
	3: & h_3 & h_1 & h_2 &\cellcolor{grey} 0 & \nicefrac{1}{3} & \cellcolor{grey} \nicefrac{2}{3} 
\end{array}
$\hspace{2em}$
\begin{array}{ccccccc}
	&&&& h_1 & h_2 & h_3 \\\cmidrule{5-7}
	1: & h_3 & h_1 & h_2 &\cellcolor{grey} \nicefrac{2}{3} & \cellcolor{grey} 0 & \cellcolor{grey} \nicefrac{1}{3} \\
	2: & h_1 & h_2 & h_3 &\cellcolor{grey} \nicefrac{1}{3} & \cellcolor{grey} \nicefrac{2}{3} & 0 \\
	3: & h_3 & h_2 & h_1 &0 & \cellcolor{grey} \nicefrac{1}{3} & \cellcolor{grey} \nicefrac{2}{3} 
\end{array}
$\hspace{2em}$
\begin{array}{ccccccc}
	&&&& h_1 & h_2 & h_3 \\\cmidrule{5-7}
	1: & h_3 & h_1 & h_2 &0 & \cellcolor{grey} 0 & \cellcolor{grey} 1 \\
	2: & h_1 & h_3 & h_2 &\cellcolor{grey} \nicefrac{1}{3} & \cellcolor{grey} \nicefrac{2}{3} & \cellcolor{grey} 0 \\
	3: & h_1 & h_2 & h_3 &\cellcolor{grey} \nicefrac{2}{3} & \cellcolor{grey} \nicefrac{1}{3} & 0 
\end{array}
$\hspace{2em}$
\begin{array}{ccccccc}
	&&&& h_1 & h_2 & h_3 \\\cmidrule{5-7}
	1: & h_3 & h_1 & h_2 &0 & \nicefrac{1}{3} & \nicefrac{2}{3} \\
	2: & h_1 & h_3 & h_2 &\cellcolor{grey} \nicefrac{1}{3} & \cellcolor{grey} \nicefrac{2}{3} & \cellcolor{grey} 0 \\
	3: & h_1 & h_3 & h_2 &\cellcolor{grey} \nicefrac{2}{3} & \cellcolor{grey} 0 & \cellcolor{grey} \nicefrac{1}{3} 
\end{array}
$\\\vspace{1em}$
\begin{array}{ccccccc}
	&&&& h_1 & h_2 & h_3 \\\cmidrule{5-7}
	1: & h_3 & h_1 & h_2 &0 & 0 & 1 \\
	2: & h_1 & h_3 & h_2 &1 & 0 & 0 \\
	3: & h_2 & h_1 & h_3 &0 & 1 & 0 
\end{array}
$\hspace{2em}$
\begin{array}{ccccccc}
	&&&& h_1 & h_2 & h_3 \\\cmidrule{5-7}
	1: & h_3 & h_1 & h_2 &0 & 0 & 1 \\
	2: & h_1 & h_3 & h_2 &1 & 0 & 0 \\
	3: & h_2 & h_3 & h_1 &0 & 1 & 0 
\end{array}
$\hspace{2em}$
\begin{array}{ccccccc}
	&&&& h_1 & h_2 & h_3 \\\cmidrule{5-7}
	1: & h_3 & h_1 & h_2 &\cellcolor{grey} 0 & \cellcolor{grey} 1 & \cellcolor{grey} 0 \\
	2: & h_1 & h_3 & h_2 &\cellcolor{grey} 1 & \cellcolor{grey} 0 & 0 \\
	3: & h_3 & h_1 & h_2 &\cellcolor{grey} 0 & \cellcolor{grey} 0 & \cellcolor{grey} 1 
\end{array}
$\hspace{2em}$
\begin{array}{ccccccc}
	&&&& h_1 & h_2 & h_3 \\\cmidrule{5-7}
	1: & h_3 & h_1 & h_2 &\cellcolor{grey} \nicefrac{2}{3} & \nicefrac{1}{3} & \cellcolor{grey} 0 \\
	2: & h_1 & h_3 & h_2 &\cellcolor{grey} \nicefrac{1}{3} & \cellcolor{grey} \nicefrac{2}{3} & 0 \\
	3: & h_3 & h_2 & h_1 &0 & \cellcolor{grey} 0 & \cellcolor{grey} 1 
\end{array}
$\\\vspace{1em}$
\begin{array}{ccccccc}
	&&&& h_1 & h_2 & h_3 \\\cmidrule{5-7}
	1: & h_3 & h_1 & h_2 &0 & 0 & 1 \\
	2: & h_2 & h_1 & h_3 &0 & 1 & 0 \\
	3: & h_1 & h_2 & h_3 &1 & 0 & 0 
\end{array}
$\hspace{2em}$
\begin{array}{ccccccc}
	&&&& h_1 & h_2 & h_3 \\\cmidrule{5-7}
	1: & h_3 & h_1 & h_2 &0 & 0 & 1 \\
	2: & h_2 & h_1 & h_3 &0 & 1 & 0 \\
	3: & h_1 & h_3 & h_2 &1 & 0 & 0 
\end{array}%
$\\\vspace{1em}

\end{document}